\documentclass[technicalreport, onecolumn, 11pt]{IEEEtran}
\IEEEoverridecommandlockouts
\usepackage{cite}
\usepackage{amsmath, amsfonts, amssymb, amsthm, mathtools}
\usepackage{algorithm, algpseudocode, calc}
\usepackage{graphicx}
\usepackage{textcomp}
\usepackage{xcolor}
\usepackage{xspace}
\usepackage{tikz}
\usepackage{hyperref}
\usepackage{ogonek}
\usepackage{comment}
\usepackage{todonotes}
\usepackage[capitalise]{cleveref}
\usepackage{calc}
\usepackage{tabularx}
\usepackage{natbib}
\usepackage{enumitem}

\newcolumntype{C}{>{\centering\arraybackslash}X}

\newtheorem{theorem}{Theorem}
\newtheorem{lemma}{Lemma}

\newtheorem{corollary}{Corollary}

\theoremstyle{definition}
\newtheorem{example}{Example}[subsection]

\DeclarePairedDelimiter\ceil{\lceil}{\rceil}
\DeclarePairedDelimiter\floor{\lfloor}{\rfloor}

\newcommand{\N}{\mathbb{N}}
\newcommand{\Q}{\mathbb{Q}}

\newcommand{\EPTAS}{\textnormal{\sffamily EPTAS}\xspace}
\newcommand{\PTAS}{\textnormal{\sffamily PTAS}\xspace}
\newcommand{\FPTAS}{\textnormal{\sffamily FPTAS}\xspace}
\newcommand{\APX}{\textnormal{\sffamily APX}\xspace}
\newcommand{\EPTASClass}{\EPTAS \xspace}

\newcommand{\APXClass}{\APX \xspace}

\newcommand{\TDM}{\textnormal{\sffamily 3-Dimensional Matching}\xspace}

\newcommand{\NPClass}{\textnormal{\sffamily NP}\xspace}
\newcommand{\NPCClass}{\textnormal{\sffamily NP-complete}\xspace}
\newcommand{\SNPCClass}{\textnormal{\sffamily Strongly NP-complete}\xspace}
\newcommand{\NPHClass}{\textnormal{\sffamily NP-hard}\xspace}
\newcommand{\SNPHClass}{\textnormal{\sffamily Strongly NP-hard}\xspace}

\newcommand{\PClass}{\textnormal{\sffamily P}\xspace}

\newcommand{\Partition}{\textnormal{\sffamily Partition}\xspace}

\newcommand{\ThreePartition}{\textnormal{\sffamily 3-Partition}\xspace}

\newcommand{\ThreeSat}{\textnormal{\sffamily 3-SAT}\xspace}

\newcommand{\MachineCovering}{\textnormal{\sffamily Machine Covering}\xspace}
\newcommand{\BoundedIS}{\textnormal{\sffamily Bounded Independent Sets}\xspace}
\newcommand{\MutualExclusion}{\textnormal{\sffamily Mutual Exclusion Scheduling}\xspace}

\newcommand{\machinepartitionrestriction}{\textit{$m$-$p$ restriction}}
\newcommand{\machinejobrestriction}{\textit{$m$-$j$ restriction}}

\newcommand{\completepartite}{\textit{complete multipartite}}

\newcommand{\kcolorable}[1]{\textit{#1-colorable}}
\newcommand{\bipartite}{\textit{bipartite}}
\newcommand{\completekpartite}[1]{\textit{complete #1-partite}}

\newcommand{\Osymbol}{\textnormal{O}}

\newcommand{\sopt}{S_{opt}}
\newcommand{\salg}{S_{alg}}
\newcommand{\fopt}{F_{opt}}
\newcommand{\cmaxcost}{C_{max}}
\newcommand{\sumcost}{\sum C_j}

\newcommand{\sv}[1][]{sv_{#1}}

\newcommand{\cover}[1][]{\ifthenelse{\equal{#1}{}}{cover\xspace}{#1-cover\xspace}}
\newcommand{\covers}[1][]{\ifthenelse{\equal{#1}{}}{covers\xspace}{#1-covers\xspace}}
\newcommand{\exactcover}{exact \cover}
\newcommand{\exactcovers}{exact \covers}
\newcommand{\tinycover}{tiny exact \cover}
\newcommand{\tinycovers}{tiny exact \covers}
\newcommand{\slackepsiloncover}{slack \cover[$(1-\epsilon)$]}
\newcommand{\slackepsiloncovers}{slack \covers[$(1-\epsilon)$]}
\newcommand{\slackexactcover}{slack \exactcover}
\newcommand{\slackexactcovers}{slack \exactcovers}
\newcommand{\epsiloncover}{\cover[$(1-\epsilon)$]}

\newcommand{\NO}{\texttt{NO}}

\newcommand{\covering}[1][]{\ifthenelse{\equal{#1}{}}{covering\xspace}{#1-covering\xspace}}
\newcommand{\coverings}[1][]{\ifthenelse{\equal{#1}{}}{coverings\xspace}{#1-coverings\xspace}}
\newcommand{\optimalcovering}{optimal \covering}
\newcommand{\exactcovering}{exact \covering}
\newcommand{\epsiloncovering}{\covering[$(1-\epsilon)$]}

\newcommand{\orderedcovering}{ordered \covering}
\newcommand{\orderedcoverings}{ordered \coverings}

\newcommand{\smallmachine}[1][]{\ifthenelse{\equal{#1}{}}{small\xspace}{#1-small\xspace}}
\newcommand{\averagemachine}[1][]{\ifthenelse{\equal{#1}{}}{average\xspace}{#1-average\xspace}}
\newcommand{\largemachine}[1][]{\ifthenelse{\equal{#1}{}}{large\xspace}{#1-large\xspace}}
\newcommand{\tinymachine}{tiny\xspace}
\newcommand{\tinyexactmachine}{tiny-exact\xspace}
\newcommand{\tinynonexactmachine}{tiny-non-exact\xspace}

\newcommand{\setssum}[3]{#1_{#2} \cup \ldots \cup #1_{#3}}
\newcommand{\setk}[1]{\{0, \ldots, #1\}}

\algnewcommand{\IfThenElse}[3]{% \IfThenElse{<if>}{<then>}{<else>}
	\State \algorithmicif\ #1\ \algorithmicthen\ #2\ \algorithmicelse\ #3}

\algnewcommand{\IfThen}[2]{% \IfThenElse{<if>}{<then>}{<else>}
	\State \algorithmicif\ #1\ \algorithmicthen\ #2\ }

\algnewcommand{\OneLineFor}[2]{% \IfThenElse{<if>}{<then>}{<else>}
	\State \algorithmicfor\ #1\ \algorithmicdo\ #2\ }

\algnewcommand{\TwoLinesFor}[2]{% 
	\State \algorithmicfor\ #1\ \algorithmicdo\ 
	\State \hspace{\algorithmicindent} #2\
}

\algnewcommand{\IfThenTwoLines}[2]{
	\State \algorithmicif\ #1\ \algorithmicthen\
	\State \hspace{\algorithmicindent}  #2\ }

\algnewcommand{\IfThenElseTwoLines}[3]{% \IfThenElse{<if>}{<then>}{<else>}
	\State \algorithmicif\ #1\ \algorithmicthen\ 
	\State \hspace{\algorithmicindent} #2\ \algorithmicelse\ #3}

\makeatletter
\newcommand{\algmargin}{\the\ALG@thistlm}
\makeatother
\algnewcommand{\parState}[1]{\State%
	\parbox[t]{\dimexpr\linewidth-\algmargin}{\strut\hangindent=\algorithmicindent \hangafter=1 #1\strut}}

\newcommand{\davg}{d_{average}}
\newcommand{\dtiny}{d_{tiny}}
\newcommand{\lmin}{l_{min}}

\Crefname{algorithm}{{$\mathbf{Algorithm}$}}{{\bfseries Algorithms}}
\creflabelformat{algorithm}{$\mathbf{#1#2#3}$}

\begin{document}
\title{Scheduling with Complete Multipartite Incompatibility Graph on Parallel Machines
{
\thanks{This work was supported by Polish National Science Center 2018/31/B/ST6/01294 grant and Gda\'nsk University of Technology, grant no. POWR.03.02.00-IP.08-00-DOK/16.}
}}

\author{\IEEEauthorblockN{Tytus Pikies}
\IEEEauthorblockA{\textit{Dept. of Algorithms and System Modeling} \\
\textit{Gda\'nsk University of Technology}\\
Gda\'nsk, Poland \\
tytpikie@pg.edu.pl}
\and \IEEEauthorblockN{Krzysztof~Turowski}
\IEEEauthorblockA{\textit{Theoretical Computer Science Dept.} \\
\textit{Jagiellonian University}\\
Krak\'ow, Poland \\
krzysztof.szymon.turowski@gmail.com}
\and \IEEEauthorblockN{Marek Kubale}
\IEEEauthorblockA{\textit{Dept. of Algorithms and System Modeling} \\
\textit{Gda\'nsk University of Technology}\\
Gda\'nsk, Poland \\
kubale@eti.pg.edu.pl}
}

\maketitle

\begin{abstract}
In this paper we consider the problem of scheduling on parallel machines with a presence of incompatibilities between jobs.
The incompatibility relation can be modeled as a complete multipartite graph in which each edge denotes a pair of jobs that cannot be scheduled on the same machine.
Our research stems from the work of Bodlaender et al. \cite{bodlaender1994scheduling,BodlaenderJOnTheComplexity1993}.
In particular, we pursue the line investigated partially by Mallek et al. \cite{mallek2019}, where the graph is complete multipartite so each machine can do jobs only from one partition.
% We also tie our results to the recent approach for so-called identical machines with class constraints by Jansen et al. \cite{jansen2020approximation}, providing a link between our case and their generalization.
% We also explain why our results seemingly might contradict the last results, but in fact this is not the case.

We provide several results concerning schedules, optimal or approximate with respect to the two most popular criteria of optimality: $\cmaxcost$ (makespan) and $\sumcost$ (total completion time).
We consider a variety of machine types in our paper: identical, uniform and unrelated.
Our results consist of delimitation of the easy (polynomial) and $\NPHClass$ problems within these constraints. 
We also provide algorithms, either polynomial exact algorithms for easy problems, or algorithms with a guaranteed constant worst-case approximation ratio or even in some cases a \PTAS for the harder ones.

In particular, we fill the gap on research for the problem of finding a schedule with the smallest $\sumcost$ on uniform machines.
We address this problem by developing a linear programming relaxation technique with an appropriate rounding, which to our knowledge is a novelty for this criterion in the considered setting.
\end{abstract}

\begin{IEEEkeywords}
job scheduling, uniform machines, makespan, total completion time, approximation schemes, NP-hardness, incompatibility graph
\end{IEEEkeywords}

\section{Introduction}

\subsection{An example application}
Imagine that we are treating some people ill with contagious diseases.
There are quarantine units containing people ill with a particular disease waiting to receive some medical services.
We also have a set of nurses. 
We would like the nurses to perform the services in a way that no nurse will travel between different quarantine units, to avoid spreading of the diseases.
Also, we would like to provide to each patient the required services, which correspond to the time to be spent by a nurse.

Consider two sample goals:
The first might be to lift the quarantine in the general as fast as possible. 
The second might be to minimize the average time of a patient waiting and treatment. %; or to minimize the cost of a treatment, when it is assumed to be proportional to total time of waiting and treatment.

The problem can be easily modeled as a scheduling problem in our model.
The jobs are the medical services to be performed.
The division of jobs into partitions of the incompatibility graph is the division of the tasks into the quarantine units.
The machines are the nurses.
The sample goals correspond to $\cmaxcost$ and $\sumcost$ criteria, respectively.

This is only a single example of an application of \emph{scheduling with incompatibility graph on parallel machines}. 
%In particular, for the sample problem it is natural to model the incompatibility as a complete multipartite graph. 

\subsection{Notation and the problems description}

We follow the notation and definitions from \cite{brucker1999scheduling}, with necessary extensions. 
Let the set of jobs be $J = \{j_1, \ldots, j_n\}$ and the set of machines be $M = \{m_1, \ldots, m_m\}$.
We denote the processing requirements of the jobs $j_1, \ldots, j_n$ as $p_1, \ldots, p_n$.

Now let us define a function $p: J \times M \rightarrow \N$, which assigns a time needed to process a given job for a given machine. 
We distinguish three main types of machines, in the ascending order of generality:
\begin{itemize}
    \item \emph{identical} -- when $p(j_i, m) = p_i$ for all $j_i \in J$, $m \in M$,
    \item \emph{uniform} -- when there exists a function $s: M \rightarrow \Q_+$, in this case $p(j_i, m) = \frac{p_i}{s(m)}$ for any $j_i \in J$, $m \in M$,
    \item \emph{unrelated} -- when there exists $s: J \times M \rightarrow \Q_+$, which assigns %a job-dependent speeds to a machine, in this case
    $p(j_i, m) = \frac{p_i}{s(j_i, m)}$ for any $j_i \in J$, $m \in M$.
\end{itemize}

The incompatibility between jobs form a relation that can be represented as a simple graph $G = (J, E)$, where $J$ is the set of jobs, and $\{j_1, j_2\}$ belongs to $E$, iff $j_1$ and $j_2$ are incompatible.
In this paper we consider complete multipartite graphs, i.e. graphs whose sets of vertices may be split into disjoint independent sets $J_1, \ldots, J_k$ (called \emph{partitions} of the graph), such that for every two vertices in different partitions there is an edge between them.
Due to the fact that the structure is simple, we omit the edges and we identify the graph with the partition of the jobs.
%Our requirement is that no machine can process jobs from two different partitions. 

We differentiate between the cases when the number of the partitions is fixed, and when it is not the case.
In the first case we denote the graph as $G = \completekpartite{k}$, and in the second as $G = \completepartite$.

A schedule $S$ is an assignment from jobs in the space of machines and starting times.
Hence if $S(j) = (m, t)$, then $j$ is executed on the machine $m$ in the time interval $[t, t + p(j, m))$ and $t + p(j, m) = C_j$ is the completion time of $j$ in $S$.
No two jobs may be executed at the same time on any machine. 
Moreover, no two jobs which are connected by an edge in the incompatibility graph may be scheduled on the same machine.
By $\cmaxcost(S)$ we denote maximum $C_j$ in $S$ over all jobs.
By $\sumcost(S)$ we denote sum of completion times of jobs in $S$.
These are two criteria of optimality of a schedule commonly considered in the literature.
Note that in both cases we are interested in finding or approximating the minimum value of respective measure.

Interestingly, an assignment from jobs in machines it sufficient to determine values of these measures in any reasonable schedule consistent with this assignment.
By reasonable we mean a schedule in which there are no unnecessary delays between processed jobs; and the jobs forming a load on any machine are in optimal order given by Smith's Rule \cite{SmithRule}, in the case of $\sumcost$ criterion.
Under such an assignment the ordering of the jobs either has no impact on the $\cmaxcost$ criterion; or is in a sense determined in the case of $\sumcost$ criterion.

We use the well-known three-field notation of \cite{lawler1982recent}.
We are interested in problems described by $\alpha|\beta|\gamma$, where
\begin{itemize}
    \item $\alpha$ is $P$ (identical machines), $Q$ (uniform machines) or $R$ (unrelated machines),
    \item $\beta$ contains either $G = \completepartite$ or $G = \completekpartite{k}$, or some additional constraints, e.g. $p_j = 1$ (unit jobs only),
    \item $\gamma$ is either $\cmaxcost$ or $\sumcost$.
\end{itemize}

\subsection{An overview of previous work}

%The scheduling problems without any restrictions on the task conflicts may be seen as having an empty conflict graph and therefore constitute the best-known special case of our problem. 
We recall that the  $P||\cmaxcost$ is \NPHClass even for two machines \cite{gareyJ1979computers}.
However, $Q||\cmaxcost$ (and therefore $P||\cmaxcost$ as well) does admit a \PTAS \cite{hochbaum1988}.
Moreover, $R_{m}||\cmaxcost$ admits a \FPTAS \cite{horowitz1976}.
There is $(2 - \frac{1}{m})$-approximation algorithm for $R||\cmaxcost$ \cite{shchepin2005optimal}; however there is no polynomial algorithm with approximation ratio better than $\frac{3}{2}$, unless $\PClass = \NPClass$ \cite{lenstra1990approximation}.
On the other hand, $Q|p_j = 1|\cmaxcost$ and $Q||\sumcost$ (with $P|p_j = 1|\cmaxcost$ and $P||\sumcost$ as their special cases) can be solved in $\Osymbol(\min\{n + m \log{m}, n \log{m}\})$ \cite{dessouky1990scheduling} and $\Osymbol(n \log{n})$ \cite[p. 133--134]{brucker1999scheduling} time, respectively.
$R||\sumcost$ can be regarded as a special case of an assignment problem \cite{bruno1974scheduling}, which can be solved in polynomial time.

The problem of scheduling with incompatible jobs for identical machines was introduced by Bodlaender et al. in \cite{bodlaender1994scheduling}.
They provided a series of polynomial time approximation algorithms for $P|G = \kcolorable{k}|\cmaxcost$.
For bipartite graphs they showed that $P|G = \bipartite|\cmaxcost$ has a polynomial $2$-approximation algorithm, and this ratio of approximation is the best possible if $P \neq NP$.
They also proved that there exist \FPTAS in the case when the number of machines is fixed and $G$ has constant treewidth.

The special case $P|G, p_j = 1|\cmaxcost$ was treated extensively in the literature under the name \BoundedIS: for given $m$ and $t$, determine whether $G$ can be partitioned into at most $t$ independent sets with at most $m$ vertices in each. 
More generally, $P|G|\cmaxcost$ is equivalent to a weighted version of \BoundedIS.
We note also that $P|G, p_j = 1|\cmaxcost$ is closely tied to \MutualExclusion, where we are looking for a schedule in which no two jobs connected by an edge in $G$ are executed in the same time. 
For the unrestricted number of machines it is the case that $P|G, p_j = 1|\cmaxcost$ has a polynomial algorithm for a certain class of graphs $G$ if and only if \MutualExclusion has a polynomial algorithm for the same class of graphs.
When all this is taken into account, there are known polynomial algorithms for solving $P|G, p_j = 1|\cmaxcost$ when $G$ is restricted to the following classes: forests \cite{baker1996mutual}, split graphs \cite{lonc1991complexity}, complements of bipartite graphs and complements of interval graphs \cite{bodlaender1995restrictions}. 
However, the problem remains \NPHClass when $G$ is restricted to bipartite graphs (even for $3$ machines), interval graphs and cographs \cite{bodlaender1995restrictions}.

Recently another line of research was established for $G$ equal to a collection of cliques (bags) in \cite{das2017minimizing}.
The authors considered $\cmaxcost$ criterion and presented a \PTAS for identical machines together with $(\log{n})^{1/4 - \epsilon}$-inapproximability result for unrelated machines. 
They also provided an $8$-approximate algorithm for unrelated machines with additional constraints. 
This approach was further pursued in \cite{grage2019eptas}, where an \EPTAS for identical machines case was presented.
The last result is a construction of \PTAS for uniform machines with some additional restrictions on machine speeds and bag sizes \cite{page2020makespan}.

Unfortunately, the case of complete multipartite incompatibility graph was not studied so extensively.
It may be inferred from \cite{bodlaender1994scheduling} that for $P|G = \completepartite|\cmaxcost$ there exists a \PTAS, which can be easily extended to $\EPTAS$; and that there is a polynomial time algorithm for $P|G = \completepartite, p_j = 1|\cmaxcost$.

% When $G$ is complete multipartite, it has to be the case that each machine serves jobs only from one partition. Therefore, it is a special case of the model, where each machine may serve jobs from $c$ different partitions. This model was investigated by Jansen et al. \cite{jansen2020approximation}, and it follows from their work that there exists a \PTAS for $P|G = \completepartite|C_{max}$.

In the case of uniform machines Mallek et al. \cite{mallek2019} proved that $Q|G = \completekpartite{2}, p_j = 1|\cmaxcost$ is \NPHClass, but it may be solved in $\Osymbol(n)$ time when the number of machines is fixed. 
Moreover, they showed an $\Osymbol(m n + m^2 \log{m})$ algorithm for the particular case $Q|G = star, p_j = 1|\cmaxcost$.
However, their result implicitly assumed that the number of jobs $n$ is encoded in binary on $\log{n}$ bits (thus making the size of schedules exponential in terms of the input size), not -- as it is customary assumed -- in unary.

In this paper we provide several results: first, we prove that $P|G = \completepartite|\sumcost$, unlike its $\cmaxcost$ counterpart, can be solved in polynomial time.
Next, we show that $Q|G = \completepartite, p_j =1|\sumcost$ is \SNPHClass and that the same holds for $\cmaxcost$ criterion.
However, it turns out that both $Q|G = \completekpartite{k}, p_j=1|\sumcost$ and $Q|G = \completekpartite{k}, p_j = 1|\cmaxcost$ admit polynomial time algorithms.
Also, we propose $2$-approximation and $4$-approximation algorithms for $Q|G = \completepartite,p_j = 1|\cmaxcost$ and $Q|G = \completepartite,p_j = 1|\sumcost$, respectively.

The first of our two main results is a $4$-approximate algorithm for $Q|G = \completekpartite{k}|\sumcost$, based on a linear programming technique.
The second one is a \PTAS for $Q|G = \completepartite, p_j=1|\cmaxcost$.

We conclude by showing that the solutions for $R|G|\cmaxcost$, or for $R|G|\sumcost$, cannot be approximated within any fixed constant, even when $G = \completekpartite{2}$ and there are only two processing times. 

\section{Identical machines}

We recall that $P|G = \completepartite|\cmaxcost$ is \NPHClass as a generalization of $P||C_{max}$, but it admits an \EPTASClass \cite{bodlaender1994scheduling}.

Focusing our attention on $\sumcost$ let us define what we mean by a \emph{greedy assignment} of machines to partitions:
\begin{enumerate}
	\item assign to each partition a single machine,
	\item assign remaining machines one by one to the partitions in a way that it decreases $\sumcost$ as much as possible.
\end{enumerate}
To see why this approach works we need the following lemma.
%We also admit that a simple dynamic programming would work for the problem, but we think that the lemma may be of some small independent interest.
%Moreover, it allows to solve the considered further special case of "restricted assignment" problem.
\begin{lemma}
	\label{lemma:nonincreasing_returns}
	 For any set of jobs, let $S_i$ be an optimal schedule in an instance of $P_i||\sumcost$ determined by $J$.
	 Then $\sumcost(S_1) - \sumcost(S_2) \ge \sumcost(S_2) - \sumcost(S_3) \ge \ldots \ge \sumcost(S_{m-1}) - \sumcost(S_{m})$.
\end{lemma}

\begin{proof}
	Assume for simplicity that $n$ is divisible by $i(i+1)(i+2)$.
	If this is not the case, then we add dummy jobs with $p_j = 0$; obviously, this does not increase $\sumcost$.
  
	Fix the ordering of jobs with respect to nonincreasing processing times.
    Now we may associate with each job its multiplier corresponding to the position in the reversed order on its machine.
    If a job $j_i$ has a multiplier $l$, then it contributes $l p_i$ to $\sumcost$, and it is scheduled as the $l$-th last job on a machine.

	Now think of the multipliers in the terms of blocks of size $i+1$.
	For $S_i$ the multipliers with respect to job order are:
	\begin{gather*}
	\underbrace{1,\ldots, 1, 1, 2}_{\text{The first block}}   ;  \underbrace{2, \ldots, 2, 3, 3}_{\text{The second block}} ; \ldots  ;  \underbrace{i, \ldots, i+1, i+1, i+1}_{\text{The (i)-th block}} ; \ldots
	%\\
	%\underbrace{i+2, \ldots, i+2, i+2, i+3}_{The \text{(i+1)-th block}} ; \ldots
	\end{gather*}
	For $S_{i + 1}$ the multipliers are:
	\begin{gather*}
	\underbrace{1,\ldots, 1, 1, 1}_{\text{The first block}}  ;  \underbrace{2, \ldots, 2, 2, 2}_{\text{The second block}} ;  \ldots   ; \underbrace{i, \ldots, i, i, i}_{\text{The (i)-th block}} ; \ldots
	%\\
	%\underbrace{i+1, \ldots, i+1, i+1, i+1}_{\text{(i+1)-th block}}  ; \ldots
	\end{gather*}
	For $S_{i + 2}$ the multipliers are:
	\begin{gather*}
	\underbrace{1, 1, 1, \ldots, 1}_{\text{The first block}}  ;  \underbrace{1, 2, 2, \ldots, 2}_{\text{The second block}} ;  \ldots  ;  \underbrace{i-1, \ldots, i-1, i, i}_{\text{The (i)-th block}} ; \ldots
	%\\
	%\underbrace{i, \ldots, i, i, i+1}_{\text{(i+1)-th block}}  \underbrace{i+1, \ldots, i+1, i+1, i+1}_{\text{(i+2)-th block}}  ; \ldots
	\end{gather*}
	Also, let the sum of multipliers of the $k$-th block in $S_{i}$ be $s_k^i$.
  
  By some algebraic manipulations we prove that
  \begin{align*}
    %s_k^i & = (i+1)k + ((k-1)\bmod i)+1 + \floor*{\frac{k-1}{i}}(i+1), \\
    s_k^i & = (i+1)k + k + \floor*{(k-1) / i}, \\
    s_k^{i+1} & = (i+1)k, \\
    %s_k^{i+2} & = (i+1)k - (k\bmod (i+2)) - \floor*{\frac{k}{i+2}}(i+1).
    s_k^{i+2} & = (i+1)k - k +  \floor*{k / (i+2)}.
  \end{align*}
  It follows directly that $s_{k-1}^i - s_{k-1}^{i+1} \ge s_{k}^{i+1} - s_{k}^{i+2}$, for $k \ge 2$.
  %as after substitution it reduces to the true inequality
  %\begin{align*}
  %  1 + (k-2) \% i -k\%(i+2) + (i+1) \big( \floor*{\frac{k-2}{i}} - \floor*{\frac{k}{i+2}} \big) \ge 0.
  %\end{align*}
	
  The smallest processing time in the $k$-th block is at least $p_{(i+1)k} \ge p_{(i+1)k + 1}$, therefore the contribution of the $k$-th block to $\sumcost(S_{i}) - \sumcost(S_{i+1})$ is at least $p_{(i+1)k + 1} (s_k^i - s_k^{i+1})$.
  Similarly, the largest processing time in the $(k+1)$-th block is at most $p_{(i+1)k + 1}$ so the contribution of the $(k+1)$-th block to $\sumcost(S_{i+1}) - \sumcost(S_{i+2})$ is at most $p_{(i+1)k+1} (s_k^{i+1} - s_k^{i+2})$.
  Thus the contribution of the $(k + 1)$-th block to $\sumcost(S_{j+1}) - \sumcost(S_{j+2})$ is at most the contribution of the $k$-th block to $\sumcost(S_j) - \sumcost(S_{j+1})$, for all $k \ge 1$.
  Also, the first block does not contribute to $\sumcost(S_{i+1}) - \sumcost(S_{i+2})$, which proves the lemma.
\end{proof}

\begin{corollary}
	For a given instance of the problem \break $P|G = \completepartite|\sumcost$ a schedule constructed by the greedy method has optimal $\sumcost$.
\end{corollary}

\begin{proof}
    Let $\salg$ and $\sopt$ be the greedy and optimal schedules, respectively.
	If the numbers of machines assigned to each of the partitions are equal in $\salg$ and $\sopt$, then the theorem obviously holds.
	
	Assume that there is a partition $J_i$ that has more machines assigned  in $\sopt$ than in $\salg$.
	It means that there is also a partition $J_j$ that has less machines assigned in $\sopt$ than in $\salg$.
	Let us construct a new schedule $\sopt$ by assigning one more machine to $J_i$ and one less to $J_j$.
	By \cref{lemma:nonincreasing_returns} we decreased $\sumcost$ on partition $J_i$ no less than we increased it on partition $J_j$.
	Hence, the claim follows.
\end{proof}

\section{Uniform machines}

It turns out that for an arbitrary number of partitions the problem is hard, even when all jobs have equal length:
\begin{theorem} 
	\label{theorem:nphcomplete}
	$Q|G = \completepartite, p_j=1|\sumcost$ is $\SNPHClass$.
\end{theorem}
\begin{proof}
    We proceed by reducing \SNPCClass \ThreePartition  \cite{gareyJ1979computers} to our problem.
    
    Recall that an instance of \ThreePartition is $(A, b, s)$, where $A$ is a set of $3 m$ elements, $b$ is a bound value, and $s$ is a size function such that for each $a \in A$, $\frac{b}{4} < s(a) < \frac{b}{2}$ and $\sum_{a \in A} s(a) = m b$.

	The question is whether $A$ can be partitioned into disjoint sets $A_1, \ldots, A_m$, such that $\forall_{1 \le i \le m} \sum_{a \in A_i} s(a) = b$.
	
	For any $(A, b, s)$ we let $G = (J_1 \cup \ldots \cup J_m, E) = \completekpartite{m}$, where $|J_i| = b$ for all $i = 1, 2, \ldots, m$.
	Moreover, let $M = \{m_1, \ldots, m_{3 m}\}$ with speeds $s(m_i) = s(a_i)$.
	Finally, let the limit value be $\sumcost = \frac{m (b + 1)}{2}$.
	
	Suppose now that an instance $(A, b, s)$ admits a $3$-partition and let the sets be $A_1$, \ldots, $A_m$. 
	Then if $a_i \in A_j$, we assign exactly $s(a_i)$ jobs from $J_j$ to the machine $m_i$.
	Since for every $i$ it holds that $\sum_{a \in A_i} s(a) = b$, we know that all jobs are assigned.
    Moreover, we never violate the incompatibility graph conditions, as we assign to any machine only jobs from a single partition.
	
	By assigning $s(a_i)$ jobs to a machine $m_i$ we ensure that
	\begin{align*}
	    \sumcost = \sum_{i = 1}^m \frac{\binom{s(a_i) + 1}{2}}{s(m_i)} = \sum_{i = 1}^m \frac{s(a_i) + 1}{2} = \frac{m (b + 1)}{2}.
	\end{align*}
	
	Conversely, suppose that we find a schedule $S$ with $\sumcost \le \frac{m (b + 1)}{2}$.
	Now let $l_i$ be the number of jobs assigned to $m_i$. 
	Let us consider the following quantity:
	\begin{align*}
	    X & := \sum_{i = 1}^m \binom{l_i + 1}{2} \frac{1}{s(m_i)} - \frac{m (b + 1)}{2} \\
	 %   & = \sum_{i = 1}^m \left(\binom{l_i + 1}{2} - \binom{s(m_i) + 1}{2}\right) \frac{1}{s(m_i)} \\
	    & = \sum_{i = 1}^m \frac{l_i + s(m_i) + 1}{s(m_i)} (l_i - s(m_i)).
	\end{align*}
	$X$ is the difference $\sumcost (S)$ and $\sumcost$ of a schedule where each machine $m_i$ is assigned $s(m_i)$ jobs.
	Now, we note that $\sum_{i = 1}^m (l_i - s(m_i)) = 0$ as every job is assigned somewhere. Moreover,
	\begin{align*}
	    l_i + s(m_i) + 1 > 2 s(m_i) & \text{\quad if $l_i \ge s(m_i)$,} \\
	    l_i + s(m_i) + 1 \le 2 s(m_i) & \text{\quad if $l_i < s(m_i)$.}
	\end{align*}
	By combining the last two facts we note that: every element $l_i - s(m_i) \ge 0$ in $X$ gets multiplied by some number greater than $2$, and every $l_i - s(m_i) < 0$ gets multiplied by some number not greater than $2$.
	Therefore $\sum_{i = 1}^m (l_i - s(m_i)) = 0$ implies $X \ge 0$.
	Moreover, if there exists any element, such that $l_i - s(m_i) > 0$, then $X > 0$.
	However, a schedule with $\sumcost \le \frac{m (b + 1)}{2}$ satisfies $X \le 0$, therefore it holds that $X = 0$ and $l_i = s(m_i)$ for all machines.
	
	Each machine has jobs from exactly one partition assigned. 
	Let $M_j$ be the set of machines on which  the jobs from $J_j$ are executed. 
	By the previous argument a machine $m_i$ has exactly $s(m_i)$ jobs assigned in $S$.
	By this and the bounds on $a \in A$, we have $|M_j| = 3$.
	Since the correctness of the schedule guarantees that all jobs are covered by some machines, we know that the division into $M_1$, $M_2$, \ldots, $M_m$ corresponds to a partition.
	%Now notice that this corresponds directly to a desired $3$-partition of $A$, since if jobs from $J_j$ are run on machines $A_j$, then $\sum_{a \in A_j} s(a) = b$.
\end{proof}
    
\begin{theorem}
	$Q|G = \completepartite, p_j=1|\cmaxcost$ is $\SNPHClass$.
\end{theorem}

\begin{proof}
    The proof is almost identical to that of \cref{theorem:nphcomplete}, only using $\cmaxcost$ as the criterion and $1$ as the bound.
\end{proof}

When the number of partitions is fixed, we show that there are polynomial algorithms for solving the respective problems.
\begin{theorem}
	There exists a $\Osymbol(m n^{k + 1} \log(m n))$ algorithm for $Q|G = \completekpartite{k}, p_j=1|\cmaxcost$.
\end{theorem}

\begin{proof}
    We adopt the Hochbaum-Shmoys framework \cite{hochbaum1988}, i.e. we guess $\cmaxcost$ of a schedule and check whether this is a feasible value.
    There are only up to $\Osymbol(m n)$ possible values of $\cmaxcost$ to consider, e.g. using binary search, as it has to be determined by the number of jobs loaded on a single machine.

    Now assume that we check a single possible value of $\cmaxcost$.
    Fix any ordering of the machines.
    We store an information, if there is a feasible assignment of the first $0 \le l \le m$ machines such that there are $a_i$ unassigned jobs for partition $J_i$.
    In each step there is a set of tuples $(a_1, \ldots, a_k)$ corresponding to the remaining jobs of $\Osymbol(n^k)$ size.
    We start with $(|J_1|, \ldots, |J_k|)$ -- if no machines are used, all jobs are unassigned.
    
    For $l \ge 1$ we take any tuple $(a_1, \ldots, a_k)$ from the previous iteration.
    We try all possible assignments of $m_l$ to the partitions, then we try all feasible assignments of the remaining jobs to the machine.
    This produces an updated tuple, determining some feasible assignment of the first $l$ machines.
    For each $(a_1, \ldots, a_k)$ we construct at most $kn$ updated tuples, each in time $\Osymbol(1)$, so the work is bounded by $\Osymbol(k n^{k + 1})$.
    Note that we do need to store only one copy of each distinct $(a_1, \ldots, a_k)$.
    
    After considering $m$ machines it is sufficient to check if the tuple $(0, 0, \ldots, 0)$ is feasible.
    Clearly the total running time for a single guess of $\cmaxcost$ is $\Osymbol(mkn^{k + 1})$.
\end{proof}

\begin{theorem}
	There exists a $\Osymbol(m n^{k + 1})$ algorithm for $Q|G = \completekpartite{k}|\sumcost$.
\end{theorem}

\begin{proof}
    Let us process machines in any fixed ordering.
    Let the state of the partial assignment be identified by a tuple $(a_1, \ldots, a_k, c)$, where $a_i$ denotes the number of vertices remaining to be covered and $c$ denotes $\sumcost$ of the jobs scheduled so far.
    
    Assume that two partial assignments $P_1$ and $P_2$ on $m'$ first machines are described by the same state $(a_1, \ldots, a_k)$; and two values $c_1$ and $c_2$, respectively.
    If $c_2 \ge c_1$, then any extension of $P_2$ on $m'' > m'$ first machines cannot be better than the exact same extension of $P_1$ on $m''$ machines.
    Therefore for any $(a_1, \ldots, a_k)$ it is sufficient to store only the tuple $(a_1, \ldots, a_k, c)$ with the smallest $c$.
    
	We may proceed with a dynamic programming similar to the one used for $\cmaxcost$.
	That is, we start with a single state $(|J_1|, \ldots, |J_k|, 0)$.
	In the $l$-th step ($l = 1, \ldots, m$) we take states $(a_1, \ldots, a_k, c)$, corresponding to feasible assignments for the first $l - 1$ machines.
	We try all $k$ possible assignments of $m_l$ to partitions.
	If $m_l$ is assigned to $J_i$, for some $i$, then we try all choices of the number $n' \in \{0, \ldots, a_i\}$ of remaining jobs from $J_i$.
	Such a choice together with the assignment of $m_l$ determines an assignment of $n'$ unassigned jobs to $m_l$.
	If the tuple constructed $(a_1' ,\ldots, a_k', c')$ has $\sumcost$ inferior to the already produced we do not store it.

    Finally, after considering all machines we obtain exactly one tuple of the form $(0, \ldots, 0, c)$, which determines the optimal schedule.
    
    At each step there are at most $n^k$ states since for every $(a_1, \ldots, a_k)$ we store only the smallest $c$.
    There are up to $n$ possible new assignments generated from each state. Each try requires $\Osymbol(k)$ time.
    Therefore, it is clear that for any fixed $k$ the time complexity of the algorithm is $\Osymbol(m n^{k + 1})$.
\end{proof}

\begin{lemma}
	\label{lemma:GeometricSumForSigma}
	Let $J$ be any set of jobs.
	Let $M$ be a set of uniform machines with respective speeds $s(m_1) = s(m_2) = 2$, $s(m_i) = 2^{i - 1}$ for $3 \le i \le m$.
	Then it holds that $\sumcost$ of an optimal schedule of $J$ on $M$ is at least as big as the optimal $\sumcost$ for $J$ and a single machine with speed $2^m$
\end{lemma}

\begin{proof}
	It is sufficient to observe that by replacing two machines $m'$ and $m''$ with $s(m') = s(m'') = 2^i$ with one machine $m$ with $s(m) = 2^{i + 1}$ we never increase the total completion time of the optimal schedule.
	
	Without loss of generality assume that there are exactly $k$ jobs assigned to both $m'$ and $m''$, as we may always add jobs with $p_j = 0$ at the beginning.
	The contribution of $m'$ and $m''$ to the total completion time is equal to
	$
	\sumcost(m', m'') = \sum_{j = 1}^{k} \frac{j p'_j}{2^i} + \sum_{j = 1}^{k} \frac{j p''_j}{2^i}
	$
	where $p'_j$ and $p''_j$ are the $j$-th last jobs scheduled on $m'$ and $m''$, respectively.
	
	%przerobic
	Now if we consider a set of machines where $m$ replaced $m'$ and $m''$ and consider a schedule where all jobs from $m'$ and $m''$ are scheduled on $m$ in an interleaving manner the first is the first one from $m'$, the second is the first one from $m''$ etc.
	Then the contribution of $m$ to the total completion time is equal to
	$
	\sumcost(m) = \sum_{j = 1}^{k} \frac{(2 j - 1) p'_j}{2^{i + 1}} + \sum_{j = 1}^{k} \frac{2 j p''_j}{2^{i + 1}} < \sumcost(m', m'').
	$
	Since this schedule is clearly a feasible one for a new set of machines we conclude that this holds also for the optimal schedule for the same set of machines.
\end{proof}

Using \cref{lemma:GeometricSumForSigma}, exhaustive search, linear programming, and rounding we were able to prove \cref{theorem:4apx}.
Roughly speaking, the main idea lies in the fact that we may guess the speeds of the machines for each of the partitions. 
Then construct a linear (possibly fractional) relaxation of the assignment of the machines to the partitions and round it to integer.
The rounding consists of rounding up the number of the machines of each speeds assigned to the partition. 
By knowledge what is the speed of fastest machine assigned to the partition in $\sopt$ and by the previous lemma, we may schedule the jobs assigned to fractions of the machines on a machine with the highest speed in the partition, increasing the total completion by at most $2$.
This together with rounding proves the following theorem. 

\begin{algorithm}
	\begin{algorithmic}[1]
		\Require $J = (J_1, \ldots, J_k)$, $M = \{m_1, \ldots, m_m\}$.
		\State Round the speeds of the machines up to the nearest multiple of $2$.
		\State Let the nonempty group of the machines, ordered by the speeds be $M_1, \ldots, M_{l}$.
		\State For each of the partitions, guess the speed of the machine with the highest speed and the number of the machines of this speed assigned.
			   Let the indices of their speed groups be $s_1', \ldots, s_k'$ and numbers be $n_1' ,\ldots, n_k'$, respectively.
			   Dismiss the guesses that are unfeasible.
		\State Solve the following linear program.
		\State Let variables be:
		\begin{itemize}
			\item $n_{pr, tp}$, where $pr \in \{1, \ldots, k\}$, $tp \in \{1, \ldots, l\}$.
			\item $x_{jb, lr, tp}$, where $jb \in J_1 \cup \ldots, \cup J_k$, $lr \in \{1, \ldots, n\}$, $tp \in \{1, \ldots, l\}$
		\end{itemize}
		\State Let the conditions be:
		\begin{align}
		\sum_{pr} n_{pr, tp} = |M_{tp}|  								   & & \forall{tp} \label{cond:MachinesInTotalRight} \\
		n_{pr, tp} = 0                                                     & & \forall{pr} \forall{tp > s_{pr}}     \label{cond:MachinesForParitionRight0}\\
		0 \le n_{pr, tp} \le |M_{tp}| - \sum_{i \in \{i | s_i = tp\}}n_i'  & & \forall{pr} \forall{tp < s_{pr}}     \label{cond:MachinesForParitionRightI}\\
		n_{pr}' = n_{pr, tp}                                               & & \forall{pr},	tp = s_{pr}             \label{cond:MachinesForParitionRightII} \\
		\sum_{lr, tp} x_{jb, lr, tp} = 1                                   & & \forall{jb}                          \label{cond:EachJobAssigned} \\
		\sum_{jb \in J_{pr}} x_{jb, lr, tp} \le n_{pr, tp}                 & & \forall{pr} \forall{tp} \forall{lr} \label{cond:NotMoreJobsInLayerThanMachines} \\
		0 \le x_{jb, lr, tp}                                               & & \forall{jb} \forall{lr} \forall{tp} \label{cond:NonNegativeJobs} \\
		0 \le n_{pr, tp}                                                   & & \forall{pr} \forall{tp}               \label{cond:NonNegativeMachines}
		%\sum_{jb \in J_{pr}} x_{jb, lr, tp} \le n_{pr, tp} & &	\forall_{pr} \forall_{tp} \forall_{lr} \label{cond:NoJobInOtherParitions}
		\end{align}
		\State Let the cost function be: $\sum_{jb, lr, tp} x_{jb, lr, tp} \cdot lr \cdot p(jb) \cdot \frac{1}{s(tp)}$, where $p(jb)$ is the processing requirement of job $jb$, $s(tp)$ is the speed factor of machine of type $tp$.
		\State Solve the jobs assignment for each partition separately using the optimal solution of LP.
	\end{algorithmic}
	\caption{$4$-approximate algorithm for the problem \break $Q|G = \completekpartite{k}|\sumcost$ }
	\label{algorithm:4ApxForSumCost}
\end{algorithm}

\begin{theorem}
	\label{theorem:4apx}
	There exists a $4$-approximation algorithm for $Q|G = \completekpartite{k}|\sumcost$.
\end{theorem}
\begin{proof}
	Consider \cref{algorithm:4ApxForSumCost}.
	First notice that the proposed program is a linear relaxation of the scheduling problem.
	Precisely, $n_{pr, tp}$ means how many machines from a group $tp$ are assigned to the partition $pr$; $x_{jb, lr, tp}$ means what part of a job $jb$ is assigned as the $lr$-th last on a machine of type $tp$.
	Notice that jobs assigned to machines of a given type form layers, i.e. jobs assigned as last contribute their processing times once, as the last by one contributes twice, etc.
	%Hence, a job $jb$ assigned to $lr$-th layer contributes $lr$ times the processing time of $jb$ times a coefficient depending on the machine speed. 
	
	About the conditions:
	\begin{itemize}
		\item Condition \eqref{cond:MachinesInTotalRight} guarantees that all the machines are assigned, fractionally at worst.
		\item Condition \eqref{cond:MachinesForParitionRight0} provides that no machine with speed higher than maximum possible (i.e. guessed) is assigned to the partition. 
		\item Condition \eqref{cond:MachinesForParitionRightI} guarantees that each of the partitions can be given any number of not preassigned (not assigned by guessing) machines of a given type.
		\item Condition \eqref{cond:MachinesForParitionRightII} guarantees that the given number of machines of guessed type are assigned to a given partition as the fastest ones
		\item Condition \eqref{cond:EachJobAssigned} ensures that any job is assigned completely, in a fractional way at worst.
		\item Condition \eqref{cond:NotMoreJobsInLayerThanMachines} guarantees that for a given layer, partition, and machine type there is no more jobs assigned than the machines of this type to the partition.
	\end{itemize}
	The cost function corresponds to an observation that a job $jb$ assigned as the $l$-th last on the machine of type $tp$ contributes exactly $\frac{l \cdot p(jb)}{s(tp)}$ to $\sumcost$.
	
	An optimal solution to LP $(x^*, n^*)$ corresponds to a fractional assignment of machines to the partitions.

	We now construct for each partition $J_i$ separately a \emph{partition fractional scheduling} problem in the following way:
	\begin{itemize}
		\item A new set of variables $y_{jb,lr,m}$ indicating a fractional assignment of $jb \in J_i$ as the $lr$-th last job on machine $m \in M'$.
		\item A cost function $\sum_{jb,lr,m} y_{jb,lr,m} \cdot \frac{lr \cdot p(jb)}{s(m)}$.
		\item The conditions:
		\begin{itemize}
		    \item $\forall_{jb, lr, m}\; y_{jb,lr,m} \ge 0$,
		    \item $\forall_{jb} \sum_{lr, m} y_{jb,lr,m} = 1$ -- each job has to be assigned completely,
		    \item $\forall_{lr, m} \sum_{jb} y_{jb, lr, m} \le 1$ -- each layer on each machine cannot contain more than a full job in total.
		\end{itemize}
	\end{itemize}
	Here the set of machines $M'$ consists of exactly $\ceil{n^*_{i, tp}}$ machines for each $1 \le tp \le l$. 
	Hence, for each type we add at most one ,,virtual'' machine due to rounding, except the machine with the highest speed per partition, which were preassigned exactly.
	
	Now we rearrange jobs within layers for machines of the same speed to construct some feasible solution to partition fractional scheduling.
	Hence, let $Y = \{y_{jb, lr, m}: s(m) = tp, jb \in J_i\}$ for any fixed $lr$, $tp$,
	We \emph{redistribute} $x_{jb, lr, tp}$ in the following way: $\forall_{lr, tp}$ we set $y_{jb, lr, m} = x_{jb, lr, tp}$ for the consecutive variables $x_{jb, lr, tp}$.
	If such an assignment would set some variable $y_{jb, lr, m} = y'$ such that $\sum_{jb}y_{jb, lr, m} > 1$, then we set $y_{jb, lr, m} = x_{jb, lr, tp} - (y'-1)$, instead.
	And we continue with next a machine of speed $tp$ and unassigned fraction of $x_{jb, lr, tp}$.
	Notice that by condition \eqref{cond:NotMoreJobsInLayerThanMachines} we have $\forall_{lr, tp} \sum_{jb}x_{jb,lr,tp}\le n^*_{i, tp} \le |Y|$.
	Since we only rearrange jobs preserving their layers the cost of $y$ in partition fractional scheduling is equal to contribution of variables from $J_i$ to the cost of $x^*$ in LP.
	Hence an optimal solution can have only at most this cost.
	
	%Moreover, if we apply this procedure to all $lr, tp$ it is clear that the resulting $y$ is a feasible solution of a partition fractional scheduling problem with a cost identical to the one found by LP.

	Let us model this linear program as a flow network. 
	Precisely, we construct:
	\begin{itemize}
		\item a set of vertices $V = J_j \cup ( M' \times |J_j|)$
		\item a set of arcs $J_j \times ( M' \times |J_j|)$ with capacity $1$ each, and with the cost of the flow by an arc $(jb, (m, lr))$ equal to $lr \cdot p(jb) \cdot \frac{1}{s(m)}$. 
	\end{itemize}
	Any fractional solution corresponds to a fractional flow by the network, i.e. a value of $y_{jb, m, lr}$ is exactly the flow by the arc $(jb, lr, m)$.
	We know that e.g. Successive Shortest Path Algorithm \cite{NetworkFlows} finds an integral minimum cost flow in the network corresponding directly to partition fractional scheduling.
	We know that this solution is at least as good as the solution for the general relaxation of the scheduling problem.
	We can treat the flow as an assignment for all the jobs in $J_i$.
	Due to rounding of $n^*$ it may use some virtual machines, but by \cref{lemma:GeometricSumForSigma} we can change them into a single virtual machine with a speed no greater than the speed of the fastest machine.
	Finally, by moving all jobs from the virtual machine and any fastest machine to this fastest machine we increase $\sumcost$ of these jobs by at most $2$ times. 
	This together with rounding of the machine speeds allows to bound the approximation ratio by $4$.
\end{proof}

	%%%%%%%%%%%%%%%%%%%%%%%%%%%%%%%%%%%%%%%%%%%%%%%%%%%%%%%%%%%%%%%%%%%%%%%%%%%%%
% Unit time jobs
%%%%%%%%%%%%%%%%%%%%%%%%%%%%%%%%%%%%%%%%%%%%%%%%%%%%%%%%%%%%%%%%%%%%%%%%%%%%%

\subsection{Simple Algorithms for Unit Time Jobs}
In this subsection we sketch outlines of two methods: a $2$-approximation algorithm for the problem $Q|G = \completepartite, p_j=1|\cmaxcost$ and a $4$-approximation algorithm for $Q|G = \completepartite, p_j=1|\sumcost$.
It is more convenient to express the algorithms in terms of covering the parts of $G$ by capacities of the machines.

Let us first sketch briefly the necessary notation.
Let $J = J_1, \ldots, J_k$ be the parts.
Let $M = m_1, \ldots, m_m$ be the machines.
Let $c: M \rightarrow \N_+$ be a function giving \emph{capacities} of the machines.

By a \emph{\cover} we mean any subset of the machines.
For \cover $M' \subseteq M$, a part $J_j$, and $c$ we define the \emph{cover ratio} $CR^c_j(M') = \frac{\min\{|J_j|, \sum_{m_i \in M'} c(m_i)\}}{|J_j|}$.
We say that a \cover $M' \subseteq M$ is an \cover[$\alpha$] for $J_j$ under $c$ if $CR^c_j(M') \ge \alpha$.
A \cover $M' \subseteq M$ is said to be an \exactcover for $J_j$ under $c$ if $CR^c_j(M') = 1$.

By a \emph{\covering} of $J' \subseteq J$ we mean a partial function $F: M \rightharpoonup J'$, i.e. a function which not necessarily assigns all the machines.
We say that a \covering $F$ of $J'$ is an \covering[$\alpha$] for $J'$ under $c$ if $\forall_{J_j \in J'}F^{-1}(J_j)$ is an \cover[$\alpha$] of $J_j$.
Similarly, a \covering $F$ of $J'$ is an \emph{\exactcovering} for $J'$ under $c$ if $\forall_{J_j \in J'}F^{-1}(J_j)$ is an \exactcover of $J_j$.

When defining particular \covers (\coverings) we usually omit the function $c$, it is clearly stated in the context of a given \cover (\covering).
Also, we sometimes do not explicitly specify $J'$ if it is clear from the context.

Now, let us proceed to the algorithms that we propose and the proofs of the quality of solutions produced by the algorithms. 
\begin{algorithm}
	\begin{algorithmic}[1]
		\Procedure{Greedy-Covering}{$J = \{J_1, \ldots, J_k\}, M = \{m_1, \ldots, m_m\}, c$}
		\State Let $F \leftarrow \emptyset$
		\State Let $j \leftarrow 1$
		\For{$i=1, \ldots, m$}
		\State $F \leftarrow F \cup (m_i, J_j)$
		\IfThen{$CR_j^c(F) \ge \frac{1}{2}$}{$j \gets j + 1$}
		\EndFor
		\IfThenElse{$F$ is \cover[$\frac{1}{2}$] for $J$}{\Return $F$;}{\Return \NO}
		\EndProcedure
	\end{algorithmic}
	\caption
	{
		Let $J$ be the parts ordered by their sizes. 
		Let $M$ be the machines ordered by their capacities $c$.
		The following procedure verifies whether there is no \exactcovering for the instance, or else it constructs a \covering[$\frac{1}{2}$].
	}
	\label{algorithm:two-covering-verifier}
\end{algorithm}

\begin{lemma}
	\label{lemma:coveringInHalf}
	For a given $(J, M, c)$ \Cref{algorithm:two-covering-verifier} either verifies that there is no \exactcovering or it constructs a \covering[$\frac{1}{2}$].
	The algorithm works in $\Osymbol(n + m)$ time.
\end{lemma}
\begin{proof} 
	First, let us assume that for the given $(J, M, c)$ there exists an \exactcovering $\fopt$.
	We establish the following invariant under this assumption: during the execution of \cref{algorithm:two-covering-verifier} after constructing a \covering[$\frac{1}{2}$] $F$ for $J_1 \ldots, J_j$ the sum of capacities of $M \setminus F^{-1}(\{J_1, \ldots, J_j\})$ at least $\sum_{j' = j + 1}^k |J_{j'}|$.
	
	First, let us prove a particular case: $\sum_{i \in [m]} c(m_i) = n$, i.e. that the instance is \emph{tight}.
	Before the assignment of the first machine the invariant obviously holds.
	Now suppose that the invariant holds for all values $1, \ldots, j - 1$ and let us use a few of the remaining unassigned machines to construct a \cover for $J_j$.
	\begin{itemize}
		\item 
		Assume that for the first machine $m_i$ assigned to $J_j$ in $F$ the following inequality holds: $c(m_i) \le |J_j|$.
		Then the invariant is preserved after the assignment, since by assigning the sequence of machines $m_i, \ldots, m_{i''}$ to $J_j$ the sum of the remaining capacities was decreased by $\sum_{i' = i}^{i''} c(m_{i'}) \le |J_j|$ and $\sum_{j' = j}^k |J_{j'}|$ was decreased by $|J_j|$.
		\item 
		Assume the opposite, that $c(m_i) > |J_j|$ holds for the first $m_i$ used for $J_j$. 
		Note that $J_j$ is covered by exactly one machine in $F$.
		Observe, the assignment of any machine $m_{i' \le i}$ to any part $J_{j' \ge j}$ in $\fopt$ would force that $\fopt^{-1}(J_j')$ would have total capacity greater than $|J_{j'}|$, which is impossible, because the instance is tight.
		This means that $\fopt^{-1}(\{J_j, \ldots, J_k\})$ consist only of machines of smaller capacity, which are in $M \setminus F^{-1}(\{J_j, \ldots, J_k\})$.
		Hence, the invariant again holds.
	\end{itemize}
	
	Finally, consider the general case $\sum_{i \in [m]} c(m_i) \ge \sum_{j \in [k]} |J_j|$.
	In such case let us consider an arbitrary \exactcovering $\fopt$.
	If there are some machines unassigned, then let us assign them arbitrary. 
	Now, $\fopt$ would be also an \exactcovering if for any $j$ the number of jobs in $J_j$ was exactly $\sum_{i: F(m_i) = J_j} c(m_i)$.
	Let us modify the instance, that is let us assume that it consists of parts $J' = \{J_1' ,\ldots, J_k'\}$ where $J_j'$ is of cardinality $\sum_{i: F(m_i) = J_j} c(m_i)$.
	Let us mark the modified instance, perhaps after resorting the parts, as $(M, (J_1', \ldots, J_k'), c)$. 
	By the observation for the tight case, the algorithm would succeed in constructing a \covering[$\frac{1}{2}$] $F'$ for $J'$. 
	Note that for any $j \in [k]$ we have $|J_j'| \ge |J_j|$, due to the fact that no part size is lower than previously. 
	Observe, by the fact that algorithm has to succeed in constructing  $F'$ it has to succeed in constructing a \covering[$\frac{1}{2}$] $F$ for $(J_1, \ldots, J_k)$.
	This is by an observation (a formal proof of which we omit) that for any $j \in [k]$ $M \setminus F'^{-1}(\{J_1', \ldots, J_j'\}) \subseteq M \setminus F^{-1}(\{J_1, \ldots, J_j\})$.
	And vice versa, if the algorithm returns \NO, there is no \exactcovering. 
\end{proof}
\begin{example}
	In order to illustrate the reasoning for the general case confer: 
	\begin{itemize}
		\item A set of machines $M = \{m_1, m_2, m_3, m_4, m_5, m_6\}$.
		\item The function $c$ with values equal to $7, 4, 4, 3, 3, 2$, for $m_1, m_2, m_3, m_4, m_5, m_6$, respectively.
		\item Two sets of parts $J' = (J_1', J_2', J_3')$, where the parts have sizes $10, 8, 5$, respectively; and $J = (J_1, J_2, J_3)$, where the parts have sizes $7, 6, 4$, respectively.
	\end{itemize}
	The \covering produced for $J'$ is $m_1; m_2, m_3; m_4$ and the \covering produced for $J$ is $m_1; m_2; m_3$.
\end{example}

%	\begin{algorithm}
%	\begin{algorithmic}[1]
%		\Procedure{Greedy-Scheduling}{$J, M, T$}
%		\IfThen{$|M| \le |J|$}{\Return \NO}
%		\State Calculate capacities $c_i = \floor{s_i \cdot T}$
%		\State Apply \cref{algorithm:two-covering-verifier} to $(J, M, c)$
%		\IfThenElseTwoLines{\cref{algorithm:two-covering-verifier} returned \NO}{\Return \NO}{\Return a schedule constructed using $F$}
%		\EndProcedure
%	\end{algorithmic}
%	\caption
%	{
%		An algorithm that either verifies that there is no schedule in time $T$; or it constructs a schedule $S$ with $\cmaxcost(S) \le 2T$.
%		The set $M$ is assumed to be sorted nonincreasingly with respect to speeds.
%		The set $J$ is assumed to be sorted nonincreasingly according to the cardinalities.
%	}
%	\label{algorithm:two-apx-cmax}
%	\end{algorithm}

\begin{theorem}
	\label{thm:Q-cmax-2apx}
	There exists $2$-approximation algorithm for $Q|G = \completepartite, p_j=1|\cmaxcost$ running in $\Osymbol(m \log m + n \log{n})$ time.
\end{theorem}
\begin{proof}
	Assume that we suspect that there exits a schedule for an instance $(J ,M)$ of $Q|G = \completepartite, p_j=1|\cmaxcost$ within time $T$.
	Let us calculate the capacities of the machines, i.e. $c(m_i) = \floor{s(m_i) \cdot T}$.
	If the assumption is true, then there exists an \exactcovering.
	Now we apply \cref{algorithm:two-covering-verifier} to $(J, M, c)$ and we may get two results:
	\begin{enumerate}
		\item the algorithm returned $F$ -- a \covering[$\frac{1}{2}$] of $J$, 
		\item or the algorithm returned \NO, which guarantees that there is no \exactcovering of $F$ in the time $T$.
	\end{enumerate}
	By \Cref{lemma:coveringInHalf} the second case cannot occur if there exists a schedule in the time $T$.
	
	Now, let us take \covering[$\frac{1}{2}$] $F$ for this $T$.
	We translate it to a schedule as follows: for $i = 1, \ldots, m$ if $F(m_i) = J_j$, we schedule up to $2 c(m_i)$ jobs from $J_j$ on $m_i$.
	In total, the space on the machines assigned to $J_j$ in the time $2T$ ensures that all jobs are scheduled.
	Moreover, each machine gets jobs only from a single part.
	Finally, it is clear that the makespan of this schedule is at most $2T$.
	
	Let $\cmaxcost^*$ be $\cmaxcost$ of an optimum schedule.
	Observe, that $\cmaxcost^*$ is determined by a number of jobs $n' \in [n]$ assigned to some machine $m_i$.
	This means that we have only $\Osymbol(mn)$ candidates for $\cmaxcost^*$.
	By checking the candidates we can find the smallest $T$ for which there exists a \covering[$\frac{1}{2}$] $F$.		
	Using $F$ it is easy to construct a schedule with $\cmaxcost$ equal to at most $2T \le 2 \cmaxcost^*$.
	
	Initially we have to sort the machines and parts, which can be done in $\Osymbol(m \log m)$ and $\Osymbol(n \log n)$ time, respectively.
	From this point that we can assume that $m \le n$; in the other case we can always discard all but $n$ fastest machines without affecting the optimal solution.
	By this we have $\Osymbol(\log{n})$ iterations of binary search over candidates for $T$.
	Each application of \cref{algorithm:two-covering-verifier} requires $\Osymbol(n)$ time, also by $m \le n$.
	Clearly, the sketched $2$-approximation algorithm requires $\Osymbol(m \log m + n \log n)$ time.
\end{proof}

\begin{theorem}
	\label{thm:Q-sumcost-4apx}
	There exists a $4$-approximation algorithm for $Q|G = \completepartite, p_j =1|\sumcost$ running in $\Osymbol(m^2n^3\log m)$ time.
\end{theorem}
\begin{proof}
	Assume that the parts and machines are sorted in order of their nonincreasing sizes and speeds, respectively.
	Suppose that we knew in advance the numbers of jobs $c_1, \ldots, c_m$ assigned to a machines $m_1, \ldots, m_m$ in some optimal schedule.
	Without loss of generality, we could assume that if $s_1 \ge \ldots \ge s_m$, then so $c_1 \ge \ldots \ge c_m$.
	Observe, that the values $c_i$ can be also interpreted to form capacities of the machines.
	In this case we could apply \cref{algorithm:two-covering-verifier} to $(J, M, c)$ and obtain $F$ -- a \covering[$\frac{1}{2}$] of $J$.
	We could translate $F$ to a schedule as follows: if $F(m_i) = J_j$, then assign up to $2c_i$ jobs from $J_j$ to $m_i$.
	In total, the capacity of the machines assigned to $J_j$ is at least $\frac{1}{2}|J_j|$ so every job would be scheduled.
	Now observe that $m_i$ in the optimal schedule contributes exactly $\binom{c_i + 1}{2} \frac{1}{s(m_i)}$ to $\sumcost$, but in constructed schedule it would contribute at most $\binom{2 c_i + 1}{2} \frac{1}{s(m_i)} \le 4 \binom{c_i + 1}{2} \frac{1}{s(m_i)}$.
	So this would be schedule with $\sumcost$ at most $4$ times the optimum.
	Unfortunately, by \cref{theorem:nphcomplete}, it is NP-hard to obtain such information.
	
	However, observe that in $F$ the assignment of the machines to the parts would be ordered.
	That is, $J_1$ gets $n_1$ machines of biggest capacity, $J_2$ gets next $n_2$ machines of biggest capacity, etc.
	This is by the observation that in \cref{algorithm:two-covering-verifier} the machines are considered in a fixed order, determined by the order of capacities, which w.l.o.g. is determined by the order of speeds.
	Let us call by \emph{\orderedcovering} any such \covering.
	
	This observation allows us to construct a \covering that corresponds to a schedule with $\sumcost$ at most $4$ times the optimum without knowledge of $c$.
	%		By an abuse of the notation we say that a \covering $F$ of $J'$ has a $\sumcost$.
	%		By this we mean that a the best with respect to $\sumcost$ schedule of $J'$ corresponding to $F$ has such a $\sumcost$. 
	We proceed by using dynamic programming over all \orderedcoverings.
	Precisely, let the states of this program be defined by $(j, i, cost, F)$.
	Where $F$ is ordered \covering in which the first $i$ machines are assigned to the first $j$ parts, and whose associated schedule has lowest $\sumcost$ among all schedules corresponding to an \orderedcovering of first $i$ machines to the first $j$ parts.
	Notice that $(1, 1, cost_1, F_1), \ldots, (1, m-(k-1), cost_{m-(k-1)}, F_{m-(k-1)})$ are well defined.
	Precisely, $cost_i$ is the total completion time of the jobs from $J_1$ scheduled on $i$ fastest machines, $F_i = \bigcup_{i' \in [i]} \{(m_{i'}, J_1)\}$.
	For $k' \ge 2$, $m-(k-k') \ge m' \ge k'$ and $m'-1 \ge m'' \ge k'-1$, construct $(k', m', cost, F')$ as the best with respect to $\sumcost$ \orderedcovering corresponding to $(k'-1, m'', cost'', F'')$ and the assignment of $m_{m''+1}, \ldots, m_{m'}$ to $J_{k'}$.
	Every such an assignment is feasible.
	
	Moreover, for any $1 \le  k' \le k$ and $k' \le m' \le m-(k'-k)$ holds that there is no \orderedcovering of $\{J_1, \ldots, J_{k'}\}$ using $\{m_1, \ldots m_{m'}\}$ with smaller $\sumcost$ of associated schedule than $cost$ following from $(k', m', cost, F)$.
	For $k' = 1$ and any $m'$ it obviously holds.
	Consider a counter-example with the minimum number of the parts and the minimum number of the machines.
	In this case, let there be an \orderedcovering $\fopt$ associated with some schedule of minimum $\sumcost$ defined by the numbers of the machines assigned to $J_1, \ldots, J_k$, and let these numbers be $n_1, \ldots, n_k$, respectively.
	Consider an \orderedcovering $F_{alg}$ determined by $(k-1, n_1 + \ldots + n_{k-1}, cost, F')$ and by the assignment of $M_k = \bigcup_{i = n_1 + \ldots + n_{k-1} + 1}^{n_1 + \ldots + n_{k}} m_i$ to $J_k$.
	Notice that the contributions to $\sumcost$ of the $J_k$ scheduled on $M_k$ are equal in the schedule associated with $\fopt$ and in the schedule associated with $F_{alg}$.
	This means that there is a minimum counter-example on $k-1$ parts and $n_1 + \ldots + n_{k-1}$ machines.
	
	The sorting of parts and machines can be done in $\Osymbol(n)$ and $\Osymbol(m \log m)$ time, respectively.
	After this operation we can assume that $m \le n$.
	If $|M| < |J|$, then no schedule can exist.
	At each step of our dynamic program there are at most $mn$ states since for every $(i, j)$ we store only the smallest $c$.
	There are up to $m$ possible new \coverings generated from each state, each requiring $\Osymbol(n \log m )$ time to generate.
	Therefore each step requires $\Osymbol(m^2n^2\log m)$ operations and the total running time of the algorithm is $\Osymbol(km^2n^2\log m) = \Osymbol(m^2n^3\log n)$.
\end{proof}

%%%%%%%%%%%%%%%%%%%%%%%%%%%%%%%%%%%%%%%%%%%%%%%%%%%%%%%%%%%%%%%%%%%%%%%%%%%%%
% A PTAS
%%%%%%%%%%%%%%%%%%%%%%%%%%%%%%%%%%%%%%%%%%%%%%%%%%%%%%%%%%%%%%%%%%%%%%%%%%%%%
\subsection{A PTAS for $Q|G = \completepartite, p_j=1|\cmaxcost$}

Now let us return to $Q|G = \completepartite, p_j=1|\cmaxcost$ problem.
We can significantly improve on \Cref{thm:Q-cmax-2apx} and construct a PTAS inspired by the ideas of the PTAS for \MachineCovering \citep{azar1998approximation}.

\begin{algorithm}
	\begin{algorithmic}[1]
		\Procedure{PTAS}{$J$, $M$, $T$,  $\epsilon$}
		\State $\epsilon \leftarrow \min\{\frac{1}{2}, \epsilon\}$
		\State Calculate rounded capacities $c^*(m)$ for all $m \in M$
		\State $l_{min} \gets \left\lceil{3 \log_{1 + \epsilon}\frac{1}{\epsilon}}\right\rceil$ + 1
		\State Split $J$ into ranges $\{P_l\}_{l = 0}^{l_{max}}$
		\State Find $SV^{l_{\min}+1}$ for $(J, M, c^*)$ (see \Cref{algorithm:prePTASCmax} and \Cref{lemma:preHeartOfPTAS})
		\For{$l = l_{min} + 1, \ldots, l_{max}$}
		\For {$sv \in SV^{l}$}
		\State Generate $CSV(sv)$ (see \Cref{algorithm:interPTASCmax,algorithm:PTASCmax}, \Cref{theorem:interPTASproperties,lemma:HeartOfPTAS-good})
		\EndFor
		\State Find $SV^{l+1}$ as a subset of $\bigcup_{sv \in SV^{l}} CSV(sv)$ (see \Cref{lem:cutting-sv})
		\EndFor
		\IfThen{$SV^{l_{max}+1} = \emptyset$}{\Return \NO}
		\State Pick any $sv \in SV^{l_{max}+1}$ with its respective \epsiloncovering $F$
		%			\State From \epsiloncovering $F$ obtain \exactcovering $F'$ by scaling $T$ by $\frac{1+\epsilon}{1-\epsilon}$.
		\State \Return a schedule $S$ constructed from $F$
		\EndProcedure
	\end{algorithmic}
	\caption{
		The main part of the PTAS for $Q|G = \completepartite, p_j=1|\cmaxcost$.
		The following algorithm either verifies that there is no schedule of length at most $T$ or it constructs a schedule with $\cmaxcost$ close to $T$.
	}
	\label{algorithm:ptas-highlevel}
\end{algorithm}

A high-level overview of the algorithm is presented as \Cref{algorithm:ptas-highlevel}.
As previously, we state the algorithm in terms of constructing an approximate \covering.
As in the case of the application of the $2$-approximation algorithm for $Q|G = \completepartite, p_j = 1|\cmaxcost$, we state the algorithm in the framework of \citet{hochbaum1988}.	
That is, we assume that a guess of a value $T$ is given such that there exists a schedule with $\cmaxcost \le T$.
Naturally, such a schedule corresponds to an \exactcovering of $J$ by $M$ under capacities given by $c(m_i) = \floor{s(m_i) \cdot T}$.
The method that we propose either verifies that the guess is incorrect, i.e. that there is no \exactcovering in capacities determined by the time $T$, or it constructs a \covering with capacities determined by the time $T$ that can be transformed into a schedule with $\cmaxcost$ near $T$.
By applying the method over a set of candidate makespans we find the smallest $T$ such that there exists a schedule with makespan close to $T$.

For any fixed guess of $T$ we can distinguish the following basic steps of the algorithm:
\begin{enumerate}
	\item Applying some preprocessing, in particular to divide parts into ranges and to calculate rounded capacities of the machines.
	\item Finding a set of vectors $SV^{l_{min}+1}$ such that at least one vector in the set describes machines that are not assigned to small parts in some \exactcovering and each vector in the set describes an \exactcovering of small parts.
	\item Applying iteratively a procedure consisting of two steps:
	\begin{enumerate}
		\item The first step is to find for each $sv^l \in SV^{l}$ a set of candidate state vectors $CSV(sv^l)$ such that it contains at least one good state vector (a term defined later) if $sv^l$ is a good state vector.
		\item The second step is to calculate $SV^{l+1}$ as a subset of $\bigcup_{sv \in SV^{l}} CSV(sv)$ such that the constructed set contains at least one good state vector.
	\end{enumerate}
	\item Constructing a schedule with $\cmaxcost \le T(1 + 7 \epsilon)$ using a nice \epsiloncovering, corresponding to a vector in $SV^{l_{max}+1}$ -- the set $SV^{l_{max}+1}$ is nonempty provided that there exists a schedule with $\cmaxcost \le T$.
\end{enumerate}

\subsubsection{Basic definitions}
In order to prove the result formally, we state a suitable notation and a few notions tailored to our problem.
As previously, for any fixed $T$ let us define the \emph{capacity} of $m \in M$ by $c(m) = \floor{s(m) \cdot T}$.
Let us also define \emph{rounded capacity} of $m$ by $c^*(m)$, equal to $c(m)$ rounded up to the nearest value of the form $\floor{(1 + \epsilon)^i}$.
Clearly, $c(m) \le c^*(m) \le (1 + \epsilon) c(m)$ so for convenience from now on we will refer to rounded capacities exclusively, and the \covers are constructed with respect to $c^*$.

Now, we group parts into sets (also called \emph{ranges}) $P_l = \{J_k\colon |J_k| \in [\floor{(1 + \epsilon)^{l}}, \floor{(1 + \epsilon)^{l + 1}})\}$ and we will consider these ranges in order of increasing $l$. 
For convenience, let $l_{max}$ be the largest value such that its range is nonempty.

Next, given $c^*$ and $l$ we divide the machines into several types:
\begin{enumerate}
	\item \emph{tiny} -- with $c^*(m_i) < \epsilon^{-2}$,
	\item \emph{small} -- with $\epsilon^{-2} \le c^*(m_i) < \epsilon (1 + \epsilon)^{l}$,
	\item \emph{average} -- with $\max\{\epsilon (1 + \epsilon)^{l}, \epsilon^{-2}\} \le c^*(m_i) < \floor{(1 + \epsilon)^{l + 1}}$,
	\item \emph{large} -- with $\max\{\floor{(1 + \epsilon)^{l + 1}}, \epsilon^{-2}\} \le c^*(m_i)$.
\end{enumerate}
The division is unambiguous only with respect to the given $l$.
For clarity of the notation we sometimes write that $m$ is \smallmachine[$l$] (\averagemachine[$l$]) /\largemachine[$l$]/ to denote that $m$ is small (average) /large/ with respect to $l$ under $c^*$.
Sometimes we do not use the $l$ explicitly when stating that some machine is small, average, etc., but it is always given implicitly. 

We use the notation of \covers and \coverings used in the description of \cref{algorithm:two-covering-verifier}.
However, we have to add a few other types of \covers and a few other types of \coverings.
For a part $J_k$ a set $M'\subseteq M$ is a \emph{\tinycover} if it is an \exactcover and $M'$ consists of tiny machines alone.
Also, we say that $M' \subseteq M$ is \emph{\slackexactcover} of a part $J_j$ in $P_l$ when it is \exactcover of $J_j$, $M'$ consists of \smallmachine[$l$] machines or tiny machines, and there is at least one \smallmachine[$l$] machine in $M'$.
Also for a \cover $M' \subseteq M$ of $J_j \in P_l$ where $M'$ consists of at least one \smallmachine[$l$] or \averagemachine[$l$] machine by \emph{slack capacity} we mean the total capacity of all \smallmachine[$l$] and \tinymachine machines in $M'$.

We define that two \covers $M'$ and $M''$ are \emph{equivalent} under capacity function $c^*$ if there is a bijection $f: M' \rightarrow M''$ such that for any $\forall_{m \in M'}c^*(m) = c^*(f(m))$.
Hence, for a set $M'$ of machines of equal capacity under $c^*$ there is $|M'|+1$ nonequivalent subsets of $M'$.

Due to the rounding we have only at most $\davg = \floor{\log_{1 + \epsilon}(\frac{1}{\epsilon})} + 1$ distinct capacities for average machines, regardless of $l$.
Their capacities are equal to $\floor{(1+ \epsilon)^{l - \davg + 1}}, \ldots, \floor{(1+\epsilon)^{l}}$ -- since is easy to check that $\floor{(1+ \epsilon)^{l - \davg}} < \epsilon (1 + \epsilon)^{l}$.
Similarly we have only at most $\dtiny = \ceil{\log_{1+\epsilon}\frac{1}{\epsilon^2}}$ distinct capacities of tiny machines (the maximum number is of the form $(1+\epsilon)^{\ceil{2\log_{1+\epsilon}\epsilon^{-1}}-1}$, the numbers are counted from $0$).
We write ``at most'' due to the fact that when $\epsilon$ is small, then a few values $(1+\epsilon)^i$ for small $i$ may be rounded to the same integer, hence there is no reason to duplicate entries.
To avoid unnecessary details we assume that there are always $\davg$ distinct capacities of average machines and $\dtiny$ distinct capacities of tiny machines.

\subsubsection{State vectors}
The crucial concept for our algorithm and its proof is the \emph{state vector} for the $l$-th range with the fields:
\[
(M_{exact}; M_{slack}, n_{small}; M_{average}; M_{large}; F),
\]
The meanings of the fields are as follows:
\begin{itemize}
	\item $M_{exact}$ -- a set of unassigned \tinymachine machines designed to form \exactcovers for some parts;
	\item $M_{slack}$ -- a set, disjoint with $M_{exact}$, of unassigned \tinymachine and \smallmachine[$l$] machines;
	\item $n_{small}$ -- the number of \smallmachine[$l$] machines in $M_{slack}$;
	\item $M_{average}$ -- a set of unassigned \averagemachine[$l$] machines;
	\item $M_{large}$ -- a set of unassigned \largemachine[$l$] machines;
	\item $F$ -- a \epsiloncovering for $P_0 \cup \ldots \cup P_{l-1}$.
\end{itemize}
First at all, the set $M_{exact}$ is necessary in our construction to guarantee that even parts of high cardinality can be covered exactly by \tinymachine machines without excessive spending of machine capacity.
Otherwise the intuition is clear; we would like to track the unassigned machines with a suitable precision: high enough that for the next ranges we can find a nice \epsiloncovering (a term defined later), but spending a polynomial amount of time.
Also, keep in mind that we would like to track the unassigned machines with respect to equivalence relation defined.
In particular, there can be $\Osymbol(m^{\dtiny})$ not equivalent sets of $M_{exact}$, $\Osymbol(m^{\davg})$ nonequivalent sets of $M_{average}$ and it is enough to consider $\Osymbol(m)$ nonequivalent sets of $M_{large}$.
The last observation is justified by the observations in the next subsection.

For clarity, for a state vector $sv^l$ we use expressions like $\sv^l.M_{exact}$ or $\sv^l.n_{small}$ to refer to the values of its fields.
Also, we denote the set of vectors as $SV^l$ to emphasize that it consists of state vectors for $l$-th range. 

\subsubsection{Good vectors and $\epsilon$-approximate coverings}
Let us consider parts in a non-decreasing order of their sizes.
Assume that there exists an \exactcovering $F$ of $J$.
Then, if $F^{-1}(J_j)$ for any $J_j \in P_l$ contains a large machine, then we may assume that $F^{-1}(J_j)$ consists of exactly one large machine $m$ -- simply we can drop other machines.
Also, we may assume that $m$ is the smallest large machine assigned to $J_j$ or a later part.
If this is not the case, then we can do as follows. 
\begin{itemize}
	\item As long as there is a \largemachine[$l$] machine $m'$ of smaller capacity that is unassigned, then exchange $m$ with $m'$.
	\item	As long as there is a \largemachine[$l$] machine $m'$ of smaller capacity used for a later part then exchange $m$ with $m'$.
\end{itemize}
Let us fix any such \exactcovering and to differentiate it from other \coverings let us denote it as the \emph{\optimalcovering}, or symbolically as $\fopt$.
Now, with respect to the \optimalcovering we can define two additional types of the machines.
For convenience w define $m$ to be \emph{\tinyexactmachine} if in the \optimalcovering $\fopt$ $m$ is assigned and the set $\fopt^{-1}(\fopt(m))$ is a \tinycover; otherwise we call $m$ \emph{\tinynonexactmachine} machine.

We use the \optimalcovering to form conditions for desirable state vectors at each step of the algorithm. 
We say that a state vector $sv^l$ is \emph{good} if:
\begin{itemize}
	\item $\sv^l.M_{exact}$ is equivalent to the \tinyexactmachine machines in $M \setminus \fopt^{-1}(P_0 \cup \ldots \cup P_{l-1})$,
	\item $\sv^l.M_{large}$ is equivalent to the \largemachine[$l$] machines in $M \setminus \fopt^{-1}(P_0 \cup \ldots \cup P_{l-1})$,
	\item $\sv^l.M_{average}$ is equivalent to the \averagemachine[$l$] machines in $M \setminus \fopt^{-1}(P_0 \cup \ldots \cup P_{l-1})$,
	\item $\sv^l.n_{small}$ is at least the number of parts in $\setssum{P}{l}{l_{max}}$ that are covered by \slackexactcover in $\fopt$ and \emph{where the fastest machine assigned in $\fopt$ is \smallmachine[$l$]},
	\item The capacity of $M_{slack}$ at least the capacity of \smallmachine[$l$] and \tinynonexactmachine machines in $M \setminus \fopt^{-1}(\setssum{P}{0}{l-1})$.
\end{itemize}
Intuitively, we would like a good state vector $\sv^l$ to describe the set of unassigned machines similar to $M \setminus \fopt^{-1}(\setssum{P}{0}{l-1})$.
%	It should be equivalent in terms of $M_{exact}$ (\tinymachine machines used exclusively for \tinycovers), \averagemachine[$l$] machines and \largemachine[$l$] machines, with at least as big capacity of $M_{slack}$ as the capacity of unassigned \smallmachine[$l$] and \tinynonexactmachine in $\fopt$, and acceptable number of \smallmachine[$l$] unassigned machines.
Most importantly, the condition on $n_{small}$, as we will see, is designed to guarantee that the number of \smallmachine[$l$] machines unassigned is at least as big as the number of parts in $\setssum{P}{l}{l_{max}}$ covered by \slackexactcover for which the fastest machine assigned in the \optimalcovering is \smallmachine[$l$].
We use the condition to guarantee that each such part can be covered by a \cover similar to \slackexactcover.

We search the space of feasible \coverings for a \epsiloncovering which is \emph{nice}. 
Formally, $M'$ is an \emph{nice \epsiloncover} of a part $J_i \in P_l$ if:
\begin{itemize}
	\item $M'$ is an \exactcover of $J_i$, i.e., $\sum_{m \in M'} c^*(m) \ge |J_i|$;
	\item or $M'$ is \emph{relatively almost \exactcover} of $J_i$, i.e., $\sum_{m \in M'} c^*(m) \ge (1 - \epsilon) |J_i|$ and $c^*(m) > \epsilon^{-1}$ for all $m \in M'$;
	\item or $M'$ is \emph{absolutely almost \exactcover} of $J_i$, i.e., $\sum_{m \in M'} c^*(m) \ge |J_i| - \epsilon^{-1}$ and $c^*(m) > \epsilon^{-2}$ for some $m \in M'$.
\end{itemize}
We call a \covering $F$ of $P_l = \{J_1, \ldots, J_{|P_l|}\}$ a \emph{nice \epsiloncovering} if for any $i \in [|P_l|]$ the set $F^{-1}(J_i)$ is nice \epsiloncover for $J_i$.
Using a nice \epsiloncovering it is easy to construct a schedule, perhaps increasing $T$ a bit.

To present the desirable property more clearly let us consider the following lemma.
\begin{lemma}
	\label{lemma:epsilontoexactcovering}
	Assume that for an instance $(J, M, c^*, \epsilon )$, where $c^*$ is an integer-valued function and $\epsilon \in (0, 1)$, there is a nice \epsiloncovering $F$.
	Then, $F$ is \exactcovering of $J$ with $M$ under $c' = \floor{c^*\left(\frac{1}{1-\epsilon} + \epsilon \right)}$.
\end{lemma}
\begin{proof}
	For $J_i \in P_l$ consider what type of cover $F^{-1}(J_i)$ is: 
	\begin{itemize}
		\item If $F^{-1}(J_j)$ is an \exactcover, then there is nothing to prove.
		\item Assume that all machines in $F^{-1}(J_i)$ have capacity at least $\epsilon^{-1}$.
		In this case the \epsiloncover has to be (at least) relatively almost \exactcover. 
		Therefore, we can multiply $c^*$ by $\frac{1}{1 - \epsilon}$ to increase the total capacity of $F^{-1}(J_i)$ by $\frac{1}{1 - \epsilon}$.
		The total capacity of $F^{-1}(J_i)$ is at least $|J_i|$ after the operation.	
		However, this could lead to some fractional capacities, so additionally we would like the capacities to be rounded up to the nearest integer.
		Summing up, we would like to replace $c^*(m)$ by $\ceil{\frac{c^*(m)}{1 - \epsilon}}$.
		This can be done e.g. by using $\floor{c^* \left(\frac{1}{1 - \epsilon} + \epsilon\right)}$ instead of $c^*$ due to the fact that
		\begin{align*}
		\ceil*{\frac{c^*(m)}{1 - \epsilon}}
		\le \floor*{\frac{c^*(m)}{1 - \epsilon}} + 1
		< \floor*{c^*(m) \left(\frac{1}{1 - \epsilon} + \epsilon\right)}
		\end{align*}
		where we used the property that $c^*(m) > \epsilon^{-1}$ for all machines used to cover $J_i$.
		\item 
		Assume that at least one machine with capacity less than $\epsilon^{-1}$ is present in $F^{-1}(J_i)$.
		However, this means that the total capacity of $F^{-1}(J_i)$ is at least equal to $|J_i| - \frac{1}{\epsilon}$, by the requirement that \epsiloncover has to be absolutely almost \exactcover in such case.
		Moreover, due to the fact that $c^*$ is integer-valued, the missing capacity has to be at most $\floor{\frac{1}{\epsilon}}$ in fact -- both $|J_i|$ and capacities are integers.
		Additionally, $F^{-1}(J_i)$ contains at least one machine $m_i$ with capacity $c^*(m_i) \ge \epsilon^{-2}$.
		Thus, it is sufficient to use $\floor{c^*(1 + \epsilon)}$ instead of $c^*$, because we get $\floor{c^*(m_i)(1 + \epsilon)} \ge \floor{c^*(m_i) + \frac{1}{\epsilon}} = c^*(m_i) + \floor{\frac{1}{\epsilon}}$ -- the additional capacity of $m_i$ brings the total capacity of cover to at least $|J_i|$.
	\end{itemize}
	Overall, it is sufficient to pick the scaling according to the worst possible case.
	Hence, every nice \epsiloncovering of $J$ with respect to $c^*$ is \exactcovering with respect to $\floor*{c^*(\frac{1}{1-\epsilon} + \epsilon)}$ % $T \left(\frac{1}{1 - \epsilon} + \epsilon\right)$.	
\end{proof}
%	Observe, that the lemma together with the proof that a nice \epsiloncovering can be efficiently calculated shows connection of the considered problem of scheduling with the problems of covering integer numbers with integer numbers.

\subsubsection{Parts with small sizes and slow machines}
Let us turn our attention to the first nontrivial part of the \Cref{algorithm:ptas-highlevel}: finding $SV^{\lmin+1}$ -- a set of state vectors for $(\lmin +1)$-th range guaranteed to contain a good state vector, under the condition that $\cmaxcost^* \le T$.
This state vector is our starting point for further iterations.

At first, let us describe the idea of \Cref{algorithm:prePTASCmax}: we split the machines into slow and fast ones and perform a dynamic programming in order to find all combinations of slow and fast machines which can be used to construct an \exactcovering of $J' = \bigcup_{i = 1}^{l_{min}} P_i$.
By the discussion on the \optimalcovering for any part from $J'$ we use only the fast machine with the smallest capacity unused yet.

\begin{algorithm}[H]
	\begin{algorithmic}[1]
		\Procedure{Find-$SV^{l_{min}+1}$}{$J$, $M$, $c^*$}
		\State Sort $J$ according to their cardinalities.
		\State $\lmin \gets \lceil 3 \log_{1+\epsilon}\frac{1}{\epsilon}\rceil + 1$
		\TwoLinesFor{$i = 0, 1, \ldots, l_{min}$}{$M_i \gets \{m\colon c^*(m) = \floor{(1 + \epsilon)^i}\}$}
		\Comment{assign each $m \in M$ to minimum feasible $M_i$}
		
		\State $J' \gets \{J_k\colon |J_k| < \floor{(1+\epsilon)^{l_{min} + 1}}\}$
		\State $M_{fast} \gets M \setminus \bigcup_{i = 1}^{l_{min}} M_i$
		\State $S_0 \gets \{(M_0, \ldots, M_{l_{min}}; M_{fast}; \emptyset)\}$
		\For{$i = 1, \ldots, |J'|$}
		\State $S_{i} \gets \emptyset$.
		\For{\textbf{each} $s = (M_0, \ldots, M_{l_{min}}; M_{fast}, F) \in S_{i-1}$}
		\For{\textbf{each} $M_0' \subseteq M_0, \ldots, M_{l_{min}}' \subseteq M_{l_{min}}$\footnotemark[2]}
		\If{$\sum_{j=0}^{l_{min}}|M_j'| \cdot \floor{(1 + \epsilon)^{j}} \ge |J_{i}|$}
		\State Let $F' := F \cup (M_0', J_i) \cup \ldots \cup (M_{l_{min}}', J_i)$
		\State Add $(M_0 \setminus M_0', \ldots, M_{l_{min}} \setminus M_{l_{min}}', M_{fast}; F')$ to $S_{i}$
		\EndIf
		\EndFor
		\If{$M_{fast} \neq \emptyset$}
		\State Let $m$ be the slowest machine in $M_{fast}$
		\State Add $(M_0, \ldots, M_{l_{min}}, M_{fast} \setminus \{m\}; F \cup (m, J_i))$ to $S_{i}$
		\EndIf
		\EndFor
		\State Trim $S_{i}$
		\EndFor
		\State $SV^{\lmin + 1} \gets \emptyset$
		\For{\textbf{each} $s = (M_0, \ldots, M_{l_{min}}; M_{fast}; F) \in S_{|J'|}$}
		\State $M_{exact} \leftarrow \emptyset$, $M_{unused} \leftarrow \emptyset$
		\For{$M_0' \subseteq M_0, \ldots, M_{\dtiny}' \subseteq  M_{\dtiny}$\footnotemark[2]}
		\OneLineFor{$i = 0, \ldots, \dtiny$}{ $M_{exact}\leftarrow M_{exact} \cup M_i'$,\;$M_{unused} \leftarrow M_{unused} \cup (M_i \setminus M_i')$}
		\OneLineFor{$i = \dtiny+1, \ldots, l_{min}$}{$M_{unused} \leftarrow M_{unused} \cup M_i$}	
		\State Transform $M_{unused}$ to $M_{slack}, M_{average}$ and $M_{large}$ for range $l_{\min}+1$ 
		\State Calculate $n_{small}$ and $c$ from $M_{slack}$
		\State Add $(M_{exact}; M_{slack}, n_{small}; M_{average}; M_{large}; F)$ to $SV^{l_{\min}+1}$
		\EndFor
		\EndFor
		%			\State Trim $SV^{\lmin + 1}$
		%			\State For each nonequivalent set $M_{exact}, M_{average}, M_{large}$ and $n_{small} \in \setk{m}$ preserve only the tuple $(M_{exact}; M_{slack}, n_{small}; M_{average}; M_{large})$ in $SV^{\lmin + 1}$ where the capacity of $M_{slack}$ is maximum.
		%			\State $SV^{\lmin + 1} \gets \emptyset$
		%			\TwoLinesFor{sets $M_{exact}$, $M_{average}$, $M_{large}$\footnotemark \; and $n_{small} \in \setk{m}$ }
		%			{Preserve only tuple $(M_{exact}; M_{slack}, n_{small}; M_{average}; M_{large})$ with the greatest capacity of $M_{slack}$ in $SV^{l_{\min}+1}$}
		\State \Return The set $SV^{l_{min}+1}$ after trimming
		\EndProcedure
	\end{algorithmic}
	\caption{An algorithm calculating the set $SV^{l_{\min}+1}$.}
	\label{algorithm:prePTASCmax}
\end{algorithm}
\footnotetext[2]{Taking into account equivalence relation.}

Formally, the properties of the algorithm are summed in the following lemma:
\begin{lemma}
	\label{lemma:preHeartOfPTAS}
	Let $l_{min} = \lceil 3 \log_{1+\epsilon}\frac{1}{\epsilon}\rceil + 1$.
	\Cref{algorithm:prePTASCmax} finds a set of state vectors $SV^{l_{min}+1}$ with the following properties:
	\begin{enumerate}[label=\textnormal{(\roman*)}]
		\item\label{preHeartOfPTAS-good} $SV^{l_{min}+1}$ contains at least one good state vector if there is an \exactcovering of $J$.
		\item\label{preHeartOfPTAS-cover} Every $sv^{l_{min}+1} \in SV^{l_{min}+1}$ contains an \exactcovering of $\setssum{P}{0}{\lmin}$.
		\item\label{preHeartOfPTAS-size} $|SV^{l_{min}+1}| = \Osymbol(m^{\dtiny + \davg + 2})$ and it can be computed in polynomial time.
	\end{enumerate}
\end{lemma}
\begin{proof}
	First at all, let us clarify the trimming operation of sets $S_{i}$ and $SV^{\lmin + 1}$.
	In the case of former set it means: for each nonequivalent set $M_0', M_1', \ldots, M_{\lmin}'; M_{large}$ preserve only one tuple in $S_i$. 
	In the case if the latter set it means: for each nonequivalent set $M_{exact}, M_{average}, M_{large}$, and a number $n_{small} \in \setk{m}$ preserve only the tuple $(M_{exact}; M_{slack}, n_{small}; M_{average}; M_{large})$ in $SV^{\lmin + 1}$ where the capacity of $M_{slack}$ is maximum, if any such tuple exists at all.
	
	Let the \optimalcovering be given as $F_{opt}$.
	Let us prove that each $S_i$ contains a tuple representing a set of machines equivalent to $M \setminus F_{opt}^{-1}(J_1 \cup \ldots \cup J_i)$.		
	This claim is obviously true for $i = 0$.
	Assume that the claim is true for $S_{i-1}$, we prove that it holds for $S_{i}$.
	Consider the tuple $s$ representing a set equivalent to $M \setminus F_{opt}^{-1}(J_1 \cup \ldots \cup J_{i-1})$.
	\begin{itemize}
		\item 
		Assume that $F_{opt}^{-1}(J_{i})$ consists of slow machines, that is, $F_{opt}^{-1}(J_{i}) \subseteq \setssum{M}{0}{l_{min}}$.
		In this case the second nested {\bfseries for} loop generates all subsets of slow machines which are \exactcover for $J_i$, hence in particular a set equivalent to $F_{opt}^{-1}(J_{i})$.
		\item 
		Assume that $F_{opt}^{-1}(J_{i})$ consists of a single fast machine.
		Observe that the {\bfseries if} generates an \exactcover by a fast machines.
	\end{itemize}	
	Thus in either case there exists a tuple $s$ in $S_i$ representing the set equivalent to $M \setminus F_{opt}^{-1}(J_1 \cup \ldots \cup J_{i})$.
	As a consequence, there exists a tuple $s \in SV^{l_{min}+1}$ that represents the machines that are equivalent to $M \setminus F_{opt}^{-1}(J')$ establishes \ref{preHeartOfPTAS-good}.
	
	An observation that for any $i \in [l_{min}+1]$ each tuple $s \in S_i$ contains an \exactcovering of $\setssum{P}{0}{i-1}$ establishes \ref{preHeartOfPTAS-cover}.
	The observation can be directly inferred from \cref{algorithm:prePTASCmax}.
	
	In order to prove \ref{preHeartOfPTAS-size} we start by noting that $|S_i| \le (m + 1)^{\lmin + 2}$ as any coordinate of any vector $s \in S_i$ can be expressed as a value from $\{0, \ldots, m\}$ (remember that we identify sets equivalent under $c^*$) and the number of coordinates of $s$ is equal to $l_{min} + 2$.
	Moreover, each $s \in S_{i - 1}$ generates at most $(m + 1)^{l_{min} + 1} + 1$ potential elements in $S_i$.
	Checking if the generated element corresponds to a feasible covering of $J_i$ and copying the \covering of $J_1, \ldots, J_i$ can be done in $\Osymbol(m)$ time.
	Thus the total time complexity of constructing $S_i$ from $S_{i - 1}$ is also $\Osymbol(m^{2\lmin + 4})$. 
	The trimming operation can be done in time $\Osymbol(m^{2\lmin + 4})$, due to the fact that the target set has size $\Osymbol((m + 1)^{\lmin + 2})$ and the number of entries that have to be visited is $\Osymbol(m^{2\lmin + 3})$ and the entries are of length $\Osymbol(m)$.
	These observations follow from the fact that $\subseteq$ is taken with respect to the equivalence relation, hence only the number of the machines taken from each group matters.
	Since $|J'| \le n$, finding $S_{|P'|}$ requires $\Osymbol(nm^{2\lmin + 4})$ time.
	
	The transformation of every $s \in S_{|P'|}$ to a set of it corresponding state vectors in $\sv^{\lmin+1}$ is simple.
	The first step is composed of two parts.
	First is to divide the set of unassigned machines into $M_{exact}$, the \tinymachine machines designed to form \exactcovers and others (there are $\Osymbol(m^{\dtiny})$ nonequivalent partitions).
	After the division, the second step is to turn the divided sets to a state vector for $\lmin+1$ directly using the definition of the state vector.
	Again the division and construction can be done in total time $\Osymbol(m^{\dtiny + 1})$.
	Hence for all the tuples this gives time $\Osymbol(m^{\lmin + \dtiny + 3})$ and the total number of tuples constructed is $\Osymbol(m^{\lmin + \dtiny + 2})$.
	After the construction the set of tupples is trimmed so that for each nonequivalent $(M_{exact}; n_{small}; M_{average}; M_{large})$ (at most $\Osymbol(m ^{\dtiny + 1 + \davg + 1})$ entries) only the entry with biggest capacity of $M_{slack}$ is preserved.
	Together this gives time $\Osymbol(n \cdot m^{2\lmin +4} + m^{\lmin  + \dtiny + 3} + m^{\dtiny + \davg + 3}) = \Osymbol(nm^{2\lmin + 4})$.
	This requires a polynomial time in $n$ and $m$ for each element of $S_{|J'|}$, thus establishing the result.
\end{proof}

\subsubsection{Finding a good state vector for range $l+1$ using a good state vector for range $l$}

Now we proceed to the essence of the algorithm: generating and merging sets of state vectors after constructing a covering of $P_l$ based on state vectors for range $l$.
During generation of sets of \emph{candidate state vectors} for every $sv^{l} \in SV^{l}$ (denoted as $CSV(sv^{l})$) three invariants are preserved:
\begin{enumerate}[label=\textnormal{(\roman*)}]
	\item If $sv^{l}$ contains a nice \epsiloncovering for $P_0 \cup \ldots \cup P_{l - 1}$, then its every candidate state vector contains a nice \epsiloncovering for $P_0 \cup \ldots \cup P_l$.
	\item If $sv^{l}$ was good, then at least one state vector among its candidate state vectors is good.
	\item For any $sv^{l} \in SV^{l}$ the cardinality of $CSV(sv^{l})$ is polynomial with respect to $n$ and $m$.
\end{enumerate}
By these invariants, it is sufficient to merge all $CSV(sv^{l})$ into $SV^{l+1}$ in such a way that if at least one $CSV(sv^{l})$ contains a good vector, then $SV^{l+1}$ also contains a good vector.

The key ideas for generation of all candidate state vectors from a given state vector $\sv^{l}$, presented in \Cref{algorithm:PTASCmax}, are as follows:
\begin{enumerate}
	\item 
	First is to consider all nonequivalent sets of \tinymachine machines reserved for \tinycovers, \averagemachine[$l$] machines, and \largemachine[$l$] machines.
	By checking all possibilities one have to match sets equivalent to the present in $\fopt^{-1}(P_l)$. 
	\item
	The second is to consider all possible values of two other numbers, checking of course whether they are compatible with $\sv^l$: 
	\begin{itemize}
		\item The number of \emph{machines} which are \emph{\smallmachine[$l$]} and \emph{used} for \slackexactcovers of parts in $P_l$ in $\fopt$ (hence everywhere below denoted as $msu$) as the fastest machines,
		\item The number of \emph{machines} which are \emph{\smallmachine[$l$]} and \emph{transferred} (hence everywhere below denoted as $mst$) and used for \slackexactcovers of parts in $\setssum{P}{l+1}{l_{max}}$ in $\fopt$ as the fastest machines.
		In the algorithm we reserve them to guarantee that some state vector constructed after constructing \covers for $P_l$ is good.
	\end{itemize}
	A correct guess of last value guarantees that if in the \optimalcovering in further ranges there is some number of parts that are covered by \slackexactcover where the fastest machine is \smallmachine[$l$], then the set of unused machines allows to construct a \slackepsiloncover for such number of parts.
	\item
	The last idea is to verify (using \Cref{algorithm:interPTASCmax}) whether the guess is feasible.
\end{enumerate}

The intuition behind \Cref{algorithm:interPTASCmax} can be summed up as the following $2$ sentences.
Assume that for a given range $P_l$ there is given amount of resources, and number of parts that have to be covered by \slackexactcover.
Then there exists an algorithm which:
\begin{itemize}
	\item either calculates a lower bound on minimum total capacity of slack machines required in any \exactcovering of $P_l$ under specified resources and conditions, and constructs a nice \epsiloncovering of $P_l$ using such an amount of slack capacity; 
	\item or it proves that no \exactcovering of $P_l$ with the provided resources exists. 
\end{itemize} 
Let us start with the description of the key procedure.

\begin{algorithm}[H]
	\algnewcommand{\LeftComment}[1]{\Statex \hspace{\algorithmicindent}\(\triangleright\) #1}
	\begin{algorithmic}[1]
		\Procedure{Check-Covering}{$M_{exact}^*$, $M_{slack}^*$, $n_{msu}^*$, $M_{average}^*$, $M_{large}^*$, $P_l = \{J_1, \ldots, J_{|P_l|}\}$}
		\State Let $M_{msu}^*$ be the set of $n_{msu}^*$ fastest \smallmachine[$l$] machines in $M_{slack}^*$ 
		\State $S_0 \leftarrow \{(M_{exact}^*, M_{slack}^* \setminus M_{msu}^*, M_{msu}^*, M_{average}^*, M_{large}^*, \emptyset)\}$
		\For {$i = 1, \ldots, |P_l|$}
		\State $S_i = \emptyset$
		\For{$s = (M_{exact}, M_{slack}, M_{msu}, M_{average}, M_{large}, F) \in S_{i-1}$}
		\LeftComment{\textbf{Case I}: \tinycover:}
		\For{\textbf{each} $M_{exact}' \subseteq M_{exact}$}
		\IfThen{$M_{exact}'$ is an \exactcover of $J_i$}{add $s \setminus M_{exact}'$\footnotemark\; to $S_i$ }
		\EndFor
		\LeftComment{\textbf{Case II}:  a single \largemachine[$l$] machine as an \exactcover:}
		\IfThen{$m \in M_{large}$}{add $s \setminus \{m\}$ to $S_i$}
		\LeftComment{\textbf{Case III}: \epsiloncover consisting of by \averagemachine[$l$] and slack machines:}
		\For{\textbf{each} nonempty $M_{average}' \subseteq M_{average}$}
		\parState{Let $M_{slack}'$ be the maximal (inclusion-wise) set of fastest machines from $M_{slack}$ such~that $\sum_{m \in M_{slack}' \cup M_{average}'} c^*(m) \le |J_i|$}
		\IfThenTwoLines{$M_{average}'\cup M_{slack}'$ is an \epsiloncover of $J_i$}{Add $s \setminus (M_{average}'\cup M_{slack}')$ to $S_i$}
		\EndFor
		\LeftComment{\textbf{Case IV}: \slackepsiloncover:}
		\State Let $m_{msu}$ be the fastest machine from $M_{msu}$
		\parState{Let $M_{slack}'$ be the maximal (inclusion-wise) set of fastest machines from $M_{slack}$ such~that $\sum_{m \in M_{slack}' \cup \{m_{msu}\}} c^*(m) \le |J_i|$}
		\IfThenTwoLines{$\{m_{msu}\} \cup M_{slack}'$ is an \epsiloncover of $J_i$}{Add $s \setminus (\{m_{msu}\} \cup M_{slack}')$ to $S_i$}
		\EndFor
		\parState{Trim $S_i$} 
		\EndFor
		\IfThenElseTwoLines{there exists $s = (\emptyset, M_{slack}^{**}, \emptyset, \emptyset, \emptyset, F)$ in $S_{|P_l|}$}{\Return $s$ where $\sum_{m \in s.M_{slack}^{**}} c^*(m)$ is maximum}{\Return \texttt{NO}}
		%		\IfThen{$S_{|P|} = \emptyset$ \textbf{or} $\forall_{s \in S_{|P_l|}} s.M_{msu} \neq \emptyset$}{\Return \texttt{NO}}
		\EndProcedure
	\end{algorithmic}
	\caption
	{
		The following algorithm either proves that there is no \exactcovering for a range $P_l$, or calculates a nice \epsiloncovering for this range.
	}
	\label{algorithm:interPTASCmax}
\end{algorithm}
\footnotetext{As a shorthand, $s \setminus M$ denotes a tuple $s = (M_{exact}, M_{slack}, M_{msu}, M_{average}, M_{large}, F)$ with machines from $M$ removed and where $\bigcup_{m \in M}(m, J_i)$ is added to $F$.}

\begin{lemma}
	Let the $l \ge l_{min}+1$-th range of parts  $P_l = (J_1, \ldots, J_{|P_l|})$ be given.
	Let a number $n_{msu}^*$ be given.
	Let also the following sets of distinct machines be given:
	\begin{itemize}
		\item $M_{exact}^*$ be a set of \tinymachine machines that can be used for \tinycover exclusively.
		\item $M_{slack}^*$ be a set of \tinymachine and \smallmachine[$l$] machines containing at least $n_{msu}^*$ \smallmachine[$l$] machines.
		\item $M_{average}^*$ be a set of \averagemachine[$l$] machines.
		\item $M_{large}^*$ be a set of \largemachine[$l$] machines.
	\end{itemize}
	Then \cref{algorithm:interPTASCmax}:
	\begin{enumerate}
		\item\label{theorem:interPTASproperties:exactcover} 
		Either determines that there is no \exactcovering $\fopt$ for $P_l$ with machines equivalent to $M_{exact}^*$ forming \tinycovers, \averagemachine[$l$] machines equivalent to $M_{average}^*$ assigned, and $|M_{large}^*|$ \largemachine[$l$] machines assigned, and assigning the amount of slack capacity bounded by the amount given in $M_{slack}^*$ in a way that $n_{msu}$ parts are covered by \slackexactcovers. 
		\item 
		Or it calculates a nice \epsiloncovering of $P_l$: with an assignment of $M_{exact}^*, M_{average}^*, M_{large}^*$, where $n_{msu}$ parts are covered by \slackepsiloncovers, and the amount of slack capacity assigned is bounded by the capacity of $M_{slack}^*$.
		Moreover, the amount of slack capacity assigned is a lower bound on the slack capacity used in any \exactcovering of $P_l$ described in (\ref{theorem:interPTASproperties:exactcover}), provided that any such \exactcovering exists.
	\end{enumerate}
	\label{theorem:interPTASproperties}
\end{lemma}
\begin{proof}
	First let us clarify what means to trim $S_i$.
	It means that the algorithm considers all nonequivalent sets $M_{exact}, M_{average}, M_{large}$ and number of machines in $M_{msu}$.
	For each unique quadruple considered it preserves only a tuple with the biggest capacity of $M_{slack}$.
	
	Assume that there exists an \exactcovering $\fopt$ of $P_l$ with the advertised properties.
	Let the set of machines assigned to $P_l$ in $\fopt$ be denoted as $M$ (keep in mind that all the machines have to be assigned).
	By a slight abuse of the notation we use the same symbols as for the \optimalcovering and for the set of machines.
	Keep in mind that the sets $M_{exact}^*$, $M_{average}^*$ and $M_{large}^*$ have equivalent sets in $M$.
	This might be not the case for $M_{slack}^*$ -- in the case of this set we are only interested in the total capacity of the machines in the set.
	Moreover, the capacity of \tinynonexactmachine and \smallmachine[$l$] machines in $M$ may be less than the capacity of $M_{slack}^*$.
	In particular, let us assume that the capacity of \tinynonexactmachine and \smallmachine[$l$] in $M$ is least possible under the specified conditions.
	
	We analyze $\fopt$ part by part considering how the set of unassigned yet machines looks like.
	To prove the theorem we prove the following invariant.
	For every $i \in [|P_l|]$ in $S_i$ there exists a tuple $s_i = (M_{exact}; M_{slack}, M_{msu}; M_{average}; M_{large}; F)$ such that:
	\begin{itemize}
		\item The set $M_{exact}$ is equivalent to the machines forming \tinycovers in $M \setminus \fopt^{-1}(\setssum{J}{1}{i})$.
		\item The set $M_{average}$ is equivalent to the average machines in $M \setminus \fopt^{-1}(\setssum{J}{1}{i})$.
		\item The number $|M_{large}|$ is exactly equal to the number of large machines in $M \setminus \fopt^{-1}(\setssum{J}{1}{i})$.
		\item $|M_{msu}|$ is exactly equal to the number of parts $J_{j > i}$ in $P_l$ that are covered by \slackexactcover in $\fopt$.
		\item The capacity of $M_{slack}^* \setminus (M_{slack} \cup M_{msu})$ (that is, the slack capacity assigned by the algorithm) is at most the slack capacity in $\fopt^{-1}(\setssum{J}{1}{i})$.
	\end{itemize}
	Moreover each $S_i$ consists of tuples containing a nice \epsiloncovering of $J_1, \ldots, J_i$.
	
	For $i = 0$ it is trivially true since \epsiloncovering constructed is empty and $M \setminus \fopt^{-1}(\emptyset) = M$.
	Assume that the invariant holds for $i-1$ and let $s_{i-1}$ be the corresponding tuple.
	Consider how $\fopt^{-1}(J_i)$ looks like:
	
	\begin{description}
		\item[(Case I)] It consists of tiny machines only.
		In this case the algorithm also constructs \exactcover for $J_i$ using a set of machines equivalent to $\fopt^{-1}(J_i)$.
		\item[(Case II)] It consists of one \largemachine[$l$] machine.
		In this case the algorithm also constructs \exactcover for $J_i$ as one large machine.
		\item[(Case III)] It consists of \averagemachine[$l$], and perhaps some \tinymachine and \smallmachine[$l$] machines.
		In this case the algorithm constructs a nice \epsiloncover for $J_i$ with an equivalent set of \averagemachine[$l$] machines and some set of \smallmachine[$l$] and \tinymachine machines.
		%			The total capacity of \tinymachine and \smallmachine[$l$] machines assigned by the algorithm cannot be greater than the total capacity of \tinymachine and \smallmachine[$l$] machines in $\fopt^{-1}(J_i)$, due to the rules of assignment of the machines (the algorithm never assigns too much capacity). 
		\item[(Case IV)] It consists of \tinymachine and \smallmachine[$l$] machines (in particular, it contains at least one \smallmachine[$l$] machine).
		The algorithm also produces a nice \epsiloncover for $J_i$ consisting of a machine from $M_{msu}$ and some \smallmachine[$l$] and \tinymachine machines.
		%			Again the total capacity of a machine from $M_{msu}$ and other \tinymachine and \smallmachine[$l$] machines assigned by the algorithm does not exceed $|J_i|$, hence it is at most equal to the capacity of \tinymachine and \smallmachine[$l$] machines in $\fopt^{-1}(J_i)$.
	\end{description}
	Observe that the algorithm produces \exactcover in the first two cases.
	In Cases III and IV it produces relatively almost \exactcover (perhaps even \exactcover) if only machines with capacity greater than $\epsilon^{-1}$ are assigned, or it produces absolutely almost \exactcover (perhaps even \exactcover) when smaller machines are assigned. 
	The fact that in Case III a nice \epsiloncovering is constructed follows from the fact that the algorithm is applied from range $P_{l_{min}+1}$.
	Recall that $l_{min} = \ceil{3 \log_{1+\epsilon}\frac{1}{\epsilon}} + 1$ and $\davg = \ceil{\log_{1+\epsilon}\frac{1}{\epsilon}} + 1$.
	Thus, for any $l > l_{min}$ we know that $\epsilon^{-2} < (1+\epsilon)^{l - \davg}$ -- hence any \averagemachine[$l$] machine has capacity greater than $\epsilon^{-2}$.
	Finally, the fact that in Case IV a nice \epsiloncovering is constructed follows from the observation that all the machines in $M_{slack}$ that are \smallmachine[$l$], hence also those reserved in $M_{msu}$ have capacity that is greater than $\epsilon^{-2}$.
	
	In the last two cases there is no waste of capacity of \tinymachine and \smallmachine[$l$] machines, due to the fact that the parts are large compared to the capacities of such machines of those types and due to the way of assigning the machines -- the algorithm never overfills the required capacity.
	Moreover in both cases the \cover can be constructed.
	That is, despite the fact that some slack is stored in $M_{msu}$ and potentially unavailable, the amount of slack is still sufficient.
	In Case III:
	\begin{itemize}
		\item Either there are no more parts covered by \slackexactcovers in $\fopt$.
		In this case $M_{slack}$ contains the necessary slack by the inductive assumption.
		\item	Or there are more parts covered by \slackepsiloncover in $\fopt$, in this case each machine in $M_{msu}$ corresponds to part that is covered by \slackexactcover later in $\fopt$.
		For each part that is covered by \slackexactcover in $\fopt$, but it is uncovered yet ``almost all'' slack capacity used to cover it in $\fopt$ is stored in $M_{slack}$.
	\end{itemize}
	Case IV is similar to Case III:
	\begin{itemize}
		\item Either $J_i$ is the last part covered by \slackexactcover in $\fopt$.
		Then, by induction the capacity in $M_{slack} \cup M_{msu}$ is sufficient to form a \cover for $J_i$.
		\item Or there are more parts covered by \slackexactcover in $\fopt$, but in this case every machine left in $M_{msu}$ certifies that there is more then enough slack.
	\end{itemize}
	In each case the desired tuple exists before the trimming and by an easy observation the trimming rule has to preserve tuple of not lower remaining slack capacity.
	Due to this, we are sure that $s_i$ has at least as much capacity present in $M_{slack} \cup M_{msu}$ as capacity of \smallmachine[$l$] and \tinynonexactmachine machines present in $M \setminus \fopt^{-1}(\{J_1, \ldots, J_i\})$.
	Therefore, the invariant holds for $i$ as well.
	
	Hence, consider the tuple $(\emptyset, M_{slack}^{**}, \emptyset, \emptyset, \emptyset, F)$, where $M_{slack}^{**}$ has maximum capacity, present after constructing \covers for all parts.
	By the invariant it has to be the case that the capacity of $M_{slack}^* \setminus M_{slack}^{**}$ is at most the slack capacity assigned to $P_l$ in $\fopt$.
\end{proof}

\begin{corollary}
	\cref{algorithm:interPTASCmax} is polynomial time.
\end{corollary}
\begin{proof}
	The first loops makes $\Osymbol(n)$ iterations.
	The second loops is over a set of $\Osymbol(m^{\dtiny + \davg + 2})$ entries.
	Then there are the following cases in parallel.
	In Case I there is loop over $\Osymbol(m^{\dtiny})$ entries.
	The next case has complexity $\Osymbol(1)$.
	In Case III there is outer loop over $\Osymbol(m^{\davg})$ entries and inner loop over $\Osymbol(m)$ entries.
	In Case IV there is only loop over $\Osymbol(m)$ entries.
	The entries have size that is $\Osymbol(m)$, due to the fact that we have to store the \covering constructed.
	The trimming rule simply proceeds over all produced entries (a set of $\Osymbol(m^{2\dtiny + 2\davg + 3})$ entries of size $\Osymbol(m)$)  and produces the set that again has $\Osymbol(m^{\dtiny + \davg + 2})$ entries.
	Together this gives time complexity $\Osymbol(nm^{2\dtiny + 2\davg + 4})$
	Hence, the algorithm is polynomial time.
\end{proof}

\begin{algorithm}[H]
	\begin{algorithmic}[1]
		\Procedure{Generate-Candidate-State-Vectors}{$\sv^l \in SV^{l}$}
		\State $CSV' \gets \emptyset$
		\For{\textbf{each} $M_{exact}^* \subseteq \sv^l.M_{exact}$, $M_{average}^* \subseteq \sv^l.M_{average}$, $M_{large}^* \subseteq \sv^l.M_{large}$}
		\For{\textbf{each} $n_{msu}^*$, $n_{mst}$}
		\IfThen{$n_{msu}^*, n_{mst}$ inconsistent with $sv^l$}{\textbf{continue}}
		\State Let $M_{mst}$ be $n_{mst}$ machines from $sv.M_{slack}$ of biggest capacity
		\State Apply \Cref{algorithm:interPTASCmax} for $(M_{exact}^*, sv.M_{slack} \setminus M_{mst}, n_{msu}^*, M_{average}^*, M_{large}^*, P_l)$
		\If{\Cref{algorithm:interPTASCmax} returned a valid $M_{slack}^{**}$ and covering $F$}
		\State $csv' \gets \sv^l \setminus M_{exact}^* \cup (\sv^l.M_{slack} \setminus (M_{slack}^{**} \cup M_{mst})) \cup M_{average}^* \cup M_{large}^*$\footnotemark
		%						\State Add $F$ to $csv'$
		\State Add $csv'$ to $CSV'$
		\EndIf
		\EndFor
		\EndFor
		\State Transform $CSV'$ to $CSV$, i.e. to state vectors for the $(l+1)$-th range
		\State \Return $CSV$
		\EndProcedure
	\end{algorithmic}
	\caption{The following algorithm generates a set of candidate state vectors for a given state vector}
	\label{algorithm:PTASCmax}
\end{algorithm}

\footnotetext{As a shorthand, $s \setminus M$ denotes a tuple $s = (M_{exact}, M_{slack}, M_{msu}, M_{average}, M_{large}, F)$ with machines from $M$ removed $s$, and where $F$ is added to $s.F$.}

\begin{lemma}
	\label{lemma:HeartOfPTAS-good}
	If $sv^{l} \in SV^{l}$ is a good state vector, then \Cref{algorithm:PTASCmax} returns a set of vectors $CSV(sv^{l})$ such that at least one $sv^{l} \in CSV(sv^{l})$ is also a good state vector.
	Moreover, if each $sv^{l}$ contains an \epsiloncovering of $\setssum{P}{0}{l}$, then each $csv^l \in CSV^{l}(sv^l)$ contains an \epsiloncovering of $\setssum{P}{0}{l+1}$.
\end{lemma} 

\begin{proof}
	Let us consider the \optimalcovering of $J$.
	Let the number of parts in $P_l$ that are covered by \slackexactcover in the \optimalcovering be exactly equal to $n_{msu}^*$.	
	There are also $n_{mst}$ parts in $\setssum{P}{l+1}{l_{max}}$ covered by \slackexactcover consisting of machines that are \tinymachine or \smallmachine[$l$].
	In particular, this means that there are $n_{mst} + n_{msu}^*$ machines in $M \setminus \fopt^{-1}(P_0 \cup \ldots \cup P_{l-1})$ that are \smallmachine[$l$].
	This is because each of the mentioned $n_{mst}$ parts can be only covered by \tinynonexactmachine and \smallmachine[$l$] machines $M'$ present in $M \setminus \fopt^{-1}(\setssum{P}{0}{l-1})$.
	This also means that the capacity of \tinynonexactmachine and \smallmachine[$l$] machines that is assigned to parts in $P_l$ in the \optimalcovering is upper bounded by the capacity of \tinynonexactmachine and \smallmachine[$l$] machines present in $(M \setminus \fopt^{-1}(\setssum{P}{0}{l-1}))$ with the capacity of $M'$ (lower bounded by $n_{mst} \cdot \floor{(1+\epsilon)^{l+1}}$) removed.
	
	\Cref{algorithm:PTASCmax} tries every possible combination of $M_{exact}^*, M_{average}^*, M_{large}^*$.
	Hence, we are certain that at some point the algorithm proceeds with sets $M_{exact}^*, M_{average}^*, M_{large}^*$ equivalent to the ones assigned to $P_l$ in the \optimalcovering.
	Moreover, since it tries every value of $n_{msu}^*$ (again from $0$ to $m$) we are certain to go through the iteration in which $n_{msu}^*$ would be equal to the number of parts in $P_l$ covered by \slackexactcover in the \optimalcovering.
	Finally, since it tries every value of $n_{mst}$ (from $0$ to $m$) we are sure that we go through an iteration such that $n_{mst}$ is exactly equal to the value derived from \optimalcovering.
	By the fact that the vector $\sv^l$ is good we have to have $\sv^l.n_{small} \ge n_{msu}^* + n_{mst}$; hence the algorithm proceeds with the values for which the inequality holds.
	Hence, the amount of slack available in $\sv^l.M_{slack} \setminus M_{mst}$ is at least $(M \setminus \fopt^{-1}(\setssum{P}{0}{l-1}))$ (because the vector is good) with the capacity of $M_{mst}$ (upper bounded by $n_{mst} \cdot \epsilon\floor{(1+\epsilon)^{l+1}}$) removed.
	Moreover, the number of \smallmachine[$l$] machines in $\sv^l.M_{slack} \setminus M_{mst}$  is at least $n_{msu}^*$, again by the fact that the vector is good.
	
	By this we know that we can apply \cref{algorithm:interPTASCmax} to construct a nice \epsiloncovering of $P_l$ and calculate a lower bound on capacity of \tinynonexactmachine and \smallmachine[$l$] machines that has to be assigned in any \exactcovering (hence in particular in the \optimalcovering) on $P_l$.
	This means that after the execution the algorithm constructs a tuple representing unassigned machines such that:
	\begin{itemize}
		\item It has equivalents set of \tinymachine machines reserved for \tinycovers to the set of \tinyexactmachine machines in $M \setminus \fopt^{-1}(P_0 \cup \ldots \cup P_l)$.
		And similarly, it has equivalents sets of \averagemachine[$l$] and \largemachine[$l$] machines.
		\item It has at least $n_{mst}$ unassigned machines that are \smallmachine[$l$].
		Moreover, the number of unassigned machines that are \smallmachine[$(l+1)$] and \averagemachine[$l$] is exactly (due to the guess) the number of machines that are \smallmachine[$(l+1)$] and \averagemachine[$l$] in $M \setminus \fopt^{-1}(\setssum{P}{0}{l})$.
		\item It has the amount of capacity remaining (in $M_{slack}^{**} \cup M_{mst}$) that is at least equal to the total capacity of \tinynonexactmachine and \smallmachine[$l$] machines in $M \setminus \fopt^{-1}(\setssum{P}{0}{l})$.
	\end{itemize}
	After the transformation of vectors the unassigned machines that are \smallmachine[$(l+1)$] but are not \smallmachine[$l$] contribute exactly the same capacity as machines from $M \setminus \fopt^{-1}(\setssum{P}{0}{l})$ that are \smallmachine[$(l+1)$] but are not \smallmachine[$l$].
	Clearly, such a tuple in $csv' \in CSV'$ is transformed to a good tuple $csv \in CSV$. 
	Each tuple in $CSV$ corresponds to an \epsiloncovering of $\setssum{P}{0}{l}$, by the properties of \cref{theorem:interPTASproperties} and the assumption that each tuple of $SV^{l}$ contains an \epsiloncovering of $\setssum{P}{0}{l-1}$.
\end{proof}
As a sidenote observe, that we did not included any trimming rules.
Hence, it follows that the sets $CSV$ may contain duplicated entries.
However, this does not matter in the light of the further considered \cref{lem:cutting-sv}.	

\begin{corollary}
	\Cref{algorithm:PTASCmax} is polynomial time and produces a set of size $\Osymbol(m^{\dtiny + \davg + 2})$.
\end{corollary}
\begin{proof}
	The first loop is over $\Osymbol(m^{\dtiny + \davg + 2})$ entries.
	The second loop is over $\Osymbol(n^2)$ entries.
	Checking whether the guess is feasible can be done in $\Osymbol(1)$ time.
	Taking subsets of $n_{mst}$ machines can be done in $\Osymbol(m)$ time.
	\Cref{algorithm:interPTASCmax} has $\Osymbol(nm^{2\dtiny + 2\davg + 4})$ time complexity.
	Construction of new tuples can be done in $\Osymbol(m)$ time.
	Transformation can be done in $\Osymbol(m)$ time.
	Together this gives $\Osymbol(n^3 m^{3\dtiny + 3\davg + 6})$ time.
	The produced set has cardinality $\Osymbol(n^2 m^{\dtiny + \davg + 2})$ entries of size $\Osymbol(m)$.
\end{proof}

Let now $SV^{l}$ for every $l = l_{min} + 1, \ldots, l_{max}$ be defined as a subset of $\bigcup_{sv \in SV^{l - 1}} CSV(sv)$ such that for every $M_{exact}$, $M_{average}$, $M_{large}$ (there are $\Osymbol(m^{\dtiny})$, $\Osymbol(m^{\davg})$, $\Osymbol(m)$ nonequivalent sets), and $n_{small} \in \setk{m}$, we keep only the state vector $(M_{exact}; M_{slack}, n_{small}; M_{average}; M_{large})$ with the biggest capacity of $M_{slack}$ machines (with ties broken arbitrarily).
\begin{lemma}
	\label{lem:cutting-sv}
	
	For any $l \in \{l_{min} + 1, \ldots, l_{max}\}$:
	\begin{enumerate}[label=\textnormal{(\roman*)}]
		\item If $SV^{l - 1}$ has at least one good state vector, then $SV^{l}$ also has at least one good state vector.
		\item If any vector in $SV^{l - 1}$ contains a nice \epsiloncovering of $P_0 \cup \ldots P_{l-1}$, then any vector in $SV^{l}$ contains a nice \epsiloncovering of $P_0 \cup \ldots \cup P_l$.
		\item The set $\bigcup_{sv \in SV^{l - 1}} CSV(sv)$ has polynomial size and can be calculated in polynomial time.
		\item $SV^{l}$ has $\Osymbol(m^{\dtiny + \davg + 2})$ state vectors for every $l = l_{min} + 1, \ldots, l_{max}$.
	\end{enumerate}
\end{lemma}

\begin{proof}
	
	In order to prove the first property observe, if $sv \in SV^{l - 1}$ was a good state vector, then $CSV(sv)$ contains at least one good state vector, by \Cref{lemma:HeartOfPTAS-good}.
	Now it is sufficient to note that for every $sv \in SV^{l - 1}$ and every good state vector $sv'$ not in $SV^{l}$ there has to be another state vector $sv'' \in SV^{l}$ with the same values of $M_{exact}$, $n_{small}$, $M_{average}$ and $M_{large}$ but at least as large capacity of $M_{slack}$ -- so $sv''$ has to be a good state vector as well.
	
	The second property follows directly from \cref{lemma:HeartOfPTAS-good}.
	
	By the construction, there is at most one state vector in $SV^{l}$ for each $M_{exact}$, $n_{small}$, $M_{average}$ and $M_{large}$.
	Thus the cardinality of $SV^{l}$ is $\Osymbol(m^{\dtiny + \davg +2})$.
	Hence, number of produced candidate state vectors is $\Osymbol(n^2m^{\dtiny + \davg + 2})$ for each vector.
	Therefore, before the trimming the produced set has cardinality $\Osymbol(n^2m^{2\dtiny + 2\davg + 4})$.
	Also, they are produced in total time $\Osymbol(m^{\dtiny + \davg + 2} \cdot n^{3}m^{3\dtiny + 3\davg + 6}) = \Osymbol(n^{3}m^{4\dtiny + 4\davg + 8})$.
	Hence, similarly to the previous cases, the trimming can be done by passing the constructed set $\bigcup_{sv \in SV^{l - 1}} CSV(sv)$ of $\Osymbol(n^2m^{2\dtiny + 2\davg + 4})$ entries of size $\Osymbol(m)$ once.
	Together this gives the total time required to produce $SV^{l}$ from $SV^{l-1}$ equal to $\Osymbol(n^{3}m^{4\dtiny + 4\davg + 8})$.
	This proves the last two points.
\end{proof}

The following conclusion follows directly from \cref{lemma:preHeartOfPTAS} and \cref{lem:cutting-sv}.
\begin{corollary}
	\label{lem:best-sv}
	If there is an \exactcovering of $J$, then $SV^{l_{max}+1}$ is nonempty.
	Any state vector from $SV^{l_{max}+1}$ contains an \epsiloncovering of $J$.
\end{corollary}

By summing all the observations up we obtain the following theorem.
\begin{theorem}
	\label{thm:killer-ptas}
	There exists a PTAS for $Q|G = \completepartite, p_j=1|\cmaxcost$.
\end{theorem}
\begin{proof}
	To avoid unnecessary details we always execute the presented algorithms with $\epsilon \le \frac{1}{2}$.
	Assume that there exists an \exactcovering within time $T$.
	Under such assumptions $SV^{l_{max}+1}$ has to contain at least one good state vector with some nice \epsiloncovering $F$, by \Cref{lem:best-sv}.
	
	Observe that if we multiply $T$ by $\left(\frac{1}{1-\epsilon} + \epsilon\right)$, then the new capacities of the form $\floor*{c^*\left(\frac{1}{1-\epsilon} + \epsilon\right)}$ guarantee that $F$ is an \exactcovering, by \cref{lemma:epsilontoexactcovering}.
	Finally, since we used the rounded capacities $c^*(m)$, we need to get back to capacities $c(m)$.
	By the fact that $c(m) \le c^*(m) \le (1 + \epsilon) c(m)$ for all $m \in M$, if there is an \exactcovering of $J$ for rounded capacities with respect to $T \left(\frac{1}{1 - \epsilon} + \epsilon\right)$, then it is \exactcovering of $J$ under the true capacities for $T (1 + \epsilon) \left(\frac{1}{1 - \epsilon} + \epsilon\right)$.
	
	To complete the proof of the approximation ratio, we note that for all $\epsilon \in (0, \frac{1}{2}]$ we have
	\begin{align*}
	(1 + \epsilon) \left(\frac{1}{1 - \epsilon} + \epsilon\right) \le (1 + \epsilon) (1 + 3 \epsilon) < (1 + 7 \epsilon),
	\end{align*}
	so our algorithm is $(1 + 7 \epsilon)$-approximation algorithm.
	
	The complexity of the algorithm is polynomial in $n$ and $m$.
	\Cref{lemma:preHeartOfPTAS} establishes that $SV^{l_{min}+1}$ can be found in time $\Osymbol(nm^{2\lmin + 4})$.
	By \Cref{theorem:interPTASproperties}, \Cref{lemma:HeartOfPTAS-good}, and \Cref{lem:cutting-sv} each $SV^l$ for $l = l_{min} + 1, \ldots, l_{max}$ and can be found in $\Osymbol(n^3m^{4\dtiny + 4\davg + 8})$ time. 
	Clearly, the number of nonempty ranges between $P_{l_{min}+1}$ and $P_{l_{max}+1}$ is at most $n$ and we can even optimize the algorithm to iterate over non empty ranges.
	Together this gives an algorithm of time complexity $\Osymbol(nm^{2\lmin + 4} + n^4m^{4\dtiny + 4\davg + 8}) = \Osymbol(n^4m^{4\dtiny + 4\davg + 8})= \Osymbol(n^4m^{12\ceil{\log_{1+\epsilon}\frac{1}{\epsilon}} + 12})$ for a given $T$.
	By combining it with a binary search over possible values of $T$ (there are $\Osymbol(mn)$ candidates in total) we complete the proof obtaining the overall time complexity $\Osymbol(\log nm \cdot n^4m^{12\ceil{\log_{1+\epsilon}\frac{1}{\epsilon}} + 12})$.
\end{proof}

\begin{example}
	\begin{table}[H]
		\centering
		\begin{tabularx}{\textwidth}{| C |c c c c || c c c c c c c |}
			\hline
			Group    & $M_{0}$ & $M_{1}$ & $M_{2}$ & $M_{3}$ & $M_{4}$ & $M_{5}$ & $M_{6}$ & $M_{7}$ & $M_{8}$ & $M_{9}$ & $M_{10}$  \\
			Capacity &$1$     & $1$     &   $2$   &  $3$    & $5$     & $7$     &  $11$   & $17$    & $25$    &   $38$  & $57$       \\
			\hline  
			Number  &$0$     & $0$     &$0$      &  $39$   & $0$     &  $2$    &   $4$   &    $2$  &    $1$  &     $0$  &     $1$  \\  
			\hline
		\end{tabularx}
		\caption
		{
			By $M_i$ we denote the set of machines of rounded capacity $\floor{(1+\epsilon)^{i}}$; here $\epsilon = 0.5$. 
			The set of tiny machine is additionally separated by a vertical line.
		}
		\label{tab:machines}
		\qquad
		\begin{tabularx}{\textwidth}{| C | c c | | c c c c | c c | }
			\hline
			Range          & $P_1$    & $P_7$      & \multicolumn{4}{ c |}{$P_8$} & \multicolumn{2}{c|}{$P_9$} \\
			Sizes in & $[1, 2)$ & $[17, 25)$ & \multicolumn{4}{ c | }{$[25,38)$} & \multicolumn{2}{c|}{$[38,57)$} \\
			Part           & $J_{1}$  & $J_{2}$    & $J_{3}$ & $J_{4}$ & $J_{5}$ & $J_{6}$ & $J_{7}$ & $J_{8}$ \\
			Size           & $3$     &   $20$     & $25$    &   $26$  &   $36$  &    $37$ & $50$    & $51$ \\  
			\hline
			The \optimalcovering    & $3$     &   $25$     & $3^9$ &   $57$  &  $3^7, 17$ & $3^9, 11$ & $3, 7^2, 11^3$ & $3^{12},17$\\ 
			\hline
			An \epsiloncovering I& $3$     &   $25$        & $57$  & $3^9$  & $11^3$         & $3, 7^2, 17$      & $3^{11}, 17$ & $ 3^{13}, 11$ \\                
			\hline
			An \epsiloncovering II& $3$     &   $25$        & $57$  & $3^9$  & $11, 17$         & $11, 17$      & $3^{8},7^2,11$ & $3^{13},11$ \\                
			\hline
			
		\end{tabularx}
		\caption
		{
			Parts, their cardinalities, a sample \exactcovering referred to as the \optimalcovering, a \epsiloncovering I corresponding to the \optimalcovering, and a \epsiloncovering II unrelated to the \optimalcovering. 
			The parts that are small, i.e. they are in the first $l_{min}$ ranges are separated by a double line.
		}
		\label{tab:parts}
		\qquad
		\begin{tabularx}{\textwidth}{| C | c c c c | c c c  | c c  | c | }
			\hline
			Multiset of machines & \multicolumn{4}{ c |}{$M_{exact}$} & \multicolumn{3}{c |}{$M_{slack}$} & \multicolumn{2}{c|}{$M_{average}$} & \multicolumn{1}{ c |}{$M_{large}$} \\
			Field             & $n_0$       & $n_1$    & $n_2$     & $n_3$    & $M_{slack}$ & $n_{small}$ & $c$ & $n_1$     & $n_2$    & $n_{large}$ \\
			\hline    
			Capacities for  $l=8$   & $1 $        & $1$      & $2$       & $3$      & $\le 11$           & -   & -   & $17$      & $25$     & $\ge 38$ \\                               
			$M \setminus \fopt^{-1}(\setssum{P}{0}{7})$           & $0$         & $0$      & $0$       & $9$      & $3^{29},7^2,11^4$ & $2$ & $145$ & $2$     & $1$    & $1$ \\
			$\sv^8$           & $0$         & $0$      & $0$       & $9$      & $3^{29},7^2,11^4$ & $6$ & $145$ & $2$     & $1$    & $1$ \\
			\hline
			Capacities for  $l=9$   & $1 $        & $1$      & $2$       & $3$      & $\le 17$                  & -    & -   & $25$      & $38$     & $\ge 57$ \\                               
			$M \setminus \fopt^{-1}(\setssum{P}{0}{8})$            & $0$         & $0$      & $0$       & $0$      & $3^{13}, 7^2, 11^3, 17$ & $2$ & $103$ & $1$   & $0$    & $0$ \\
			$\sv[a]^9$           & $0$         & $0$      & $0$       & $0$      & $3^{28}, 11, 17$        & $2$ & $112$ & $1$   & $0$    & $0$ \\  
			$\sv[b]^9$           & $0$         & $0$      & $0$       & $0$      & $3^{28}, 7^2, 11^2$        & $4$ & $120$ & $1$   & $0$    & $0$ \\  
			\hline
		\end{tabularx}
		\caption
		{
			Here the sets $M_{exact}, M_{average}$ and $M_{large}$ are represented by numbers of machines of given capacity -- due to the observation of equivalence of machines under given capacity.
			The ``vectors'' $M \setminus \fopt^{-1}(\setssum{P}{0}{7})$ and $M \setminus \fopt^{-1}(\setssum{P}{0}{8})$ were constructed for convenience, they describe the \optimalcovering presented in \cref{tab:parts}.
			That is, they characterize exactly $M_{exact}$, $M_{average}$ and $M_{large}$ of a good vector and give lower bounds on $n_{small}$ and capacity of $M_{slack}$. 
			$n_{small}$ for $\fopt$ represents the $2$ parts in ranges $P_8$ or later for which a \smallmachine[$8$] machine (i.e. small from perspective of the $8$-th range) is the fastest machine.
			Notice that $\fopt^{-1}(J_8)$ is \slackexactcover, but the fastest machine is \smallmachine[$9$] but not \smallmachine[$8$]. 
			Finally, observe that both $\sv[a]^9$ and $\sv[b]^9$ are good, despite the fact that they have fewer \smallmachine[$9$] machines available than there is in $M \setminus \fopt^{-1}(\setssum{P}{0}{8})$.
			Moreover, during construction of $\sv[b]^9$ both $n_{msu}^*$ and $n_{mst}$ were guessed incorrectly; despite this a good vector was constructed.
			However, only the existence of $\sv[a]^9$ is guaranteed by \cref{lemma:HeartOfPTAS-good}.
		}
		\label{tab:bv}
	\end{table}
	
	As an example confer the data given by:
	\begin{itemize}
		\item A precision parameter $\epsilon = 0.5$.
		\item Guessed $T = 1$, determining the presented capacities.
		\item A set of machines given in \cref{tab:machines}.
		The machines are grouped into sets of respective cardinalities.
		To avoid introducing excessive number of symbols we identify the machines with their capacities and we will refer to the machines by the numbers only.
		\item A set of parts grouped into ranges given in \cref{tab:parts}.
		\item A sample \exactcovering chosen to be the \optimalcovering $\fopt$. 
		\item Good vectors $\sv^8$, $\sv[a]^9$, and $\sv[b]^9$ given in \cref{tab:bv}.
	\end{itemize} 
	
	Observe that here $l_{min} = 7$, hence $\setssum{P}{0}{7}$ are covered by \exactcovers using \cref{algorithm:prePTASCmax} and there is a, potentially good, vector $\sv^8$ constructed.
	
	As presented in \cref{tab:parts}, notice that $\fopt^{-1}(J_3), \fopt^{-1}(J_4), \fopt^{-1}(J_5), \fopt^{-1}(J_6)$ are: a set of tiny machines, a \largemachine[$8$] machine, \averagemachine[$8$], \tinynonexactmachine, and \smallmachine[$l$] machines, and \tinymachine and \smallmachine[$8$] machines, respectively.
	Moreover $\fopt^{-1}(J_7)$ and $\fopt^{-1}(J_8)$ are \slackexactcovers.
	However, in $\fopt^{-1}(J_7)$ the fastest machine is \smallmachine[$8$].
	This means that at least one \smallmachine[$8$] is in $M \setminus \fopt^{-1}(\setssum{P}{0}{8})$.
	Moreover, it means that the capacity of \tinynonexactmachine and \smallmachine[$8$] machines in $M \setminus \fopt^{-1}(\setssum{P}{0}{8})$ is at least $38$, by the definition of range.
	
	Hence, in at least one iteration \cref{algorithm:PTASCmax} considers a good state vector; for example $\sv^8$, presented in \cref{tab:parts}.
	Moreover, \tinymachine exact, \averagemachine[$8$], and \largemachine[$8$] machines that are to be assigned to parts in $P_8$ are equivalent to the assigned in the \optimalcovering; in the example $3^9, 17, 57$, respectively.	
	The algorithm guesses the value $n_{mst}$, the number of parts covered by \slackexactcover in $P_9$ or later range (here only $P_9$) \emph{where the fastest machines are \smallmachine[$8$] machines}; here it is equal to $1$.
	Moreover, the number of parts in $P_8$ that are covered by \slackexactcover are guessed; in the example $n_{msu} = 1$.
	Observe that using the guesses the algorithm  constructs the set $M_{mst}$ (in the example $M_{mst} = \{11\}$).
	The amount of capacity reserved in $M_{mst}$ is much less than the capacity of \tinynonexactmachine \smallmachine[$8$] machines assigned in $\fopt^{-1}(P_9)$ 
	
	Hence, \cref{algorithm:interPTASCmax} in at least one iteration is applied with equivalent set of \tinymachine machines reserved for \tinycovers, \largemachine[$8$], and \averagemachine[$8$] machines, a sufficient number of \smallmachine[$8$] machines, equal at least to $n_{msu}$ and at least the same amount of capacity in $M_{msu}$ as the capacity of \tinynonexactmachine and \smallmachine[$l$] in $\fopt^{-1}(P_8)$.
	Therefore, \cref{algorithm:PTASCmax} has to return a good vector, perhaps $\sv[a]^9$.
	Observe, that $\sv[a]^9$ may perhaps have fewer \smallmachine[$9$] machines than present in $M \setminus \fopt^{-1}(\setssum{P}{0}{8})$.
	However, by reserving $M_{mst}$ and guessing the \averagemachine[$8$] and \smallmachine[$9$] machines in $M \setminus \fopt^{-1}(\setssum{P}{0}{8})$, still the number of \smallmachine[$9$] machines is enough to construct \slackepsiloncovers for every part covered by \slackexactcovers in \optimalcovering in $\setssum{P}{9}{l_{max}}$.
	Of course, during the execution it might be the case that a superior good state vector is constructed, for example $\sv[b]^9$, even from guesses not corresponding to the \optimalcovering.
	By \cref{theorem:interPTASproperties} it might be used as well to construct an \epsiloncovering in further ranges.
	
	Considering the constructed \epsiloncovering, observe that $J_5$ is covered by \slackepsiloncover and it is relatively almost cover.
	Notice that $J_7, J_8$ are covered by \slackepsiloncovers and they are absolutely almost covers.
\end{example}

\section{Unrelated machines}

We prove that there is no good approximation algorithm possible in the case of unrelated machines.
\begin{theorem}
    There is no constant approximation ratio algorithm for $R|G = \completekpartite{2}|\sumcost$ \break $(R|G = \completekpartite{2}|\cmaxcost)$. % and \break $R|G = \completekpartite{2}|\sumcost$.
\end{theorem}

\begin{proof}
	Assume that there is $d$-approximate algorithm for the $R|G = \completekpartite{2}|\sumcost$ problem ($R|G = \completekpartite{2}|\cmaxcost$).
	Consider an instance of \ThreeSat with the set of variables $V$ and the set of clauses $C$, moreover where for each $v \in V$ there are at most $5$ clauses containing $v$.
	This version is still \NPCClass \cite{gareyJ1979computers}.
	We construct the scheduling instance as follows: let $M = \{v^T, v^F: v \in V\}$.
	Let also $G = \completekpartite{2}$ with partitions $J_1 = \{j_{v, 1}: v \in V\} \cup \{j_c: c \in C\}$ and $J_2 = \{j_{v, 2}: v \in V\}$.
	Hence $n = 2|V| + |C| \le 7|V|$, by $|C| \le 5|V|$.
	
	Let $p_j = 1$ for all jobs.
	Let $s_1 \ge 1$ be a value determined by an instance of \ThreeSat, but polynomially bound by the size of the instance.
	Let now $s(j_{v, 1}, v^T) = s(j_{v, 2}, v^T) = s(j_{v, 1}, v^F) = s(j_{v, 2}, v^F) = s_1$, for any $v \in V$.
	Let for any $c \in C$ $s(j_c, v^T) = s_1$ if $v$ appears in $c$, and $s(j_c, v^F) = s_1$ if $\neg{v}$ appears in $c$.
	%For example, for a clause $c = (v_1, v_2, \neg{v_4})$ a job $j_c$ is added to $J_1$, for which $s(j_c, v^T_1) = s(j_c, v^T_2) = s(j_c, v^F_4) = s_1$.
	Set all other $s(j, m)$ to $1$.
	
	%Consider any schedule.
	%If, some machines $v^T$ and $v^F$ are in the same partition, then at least one job is processed with time $1$.
	Consider any instance of the scheduling problem, corresponding to an instance of \ThreeSat with answer \texttt{YES}.
	Then we can schedule $J$ on the machines according to fulfilling valuation, in the following way:
	If $v$ has value $T$ then we assign $v^T$ to $J_1$ and $v^F$ to $J_2$, otherwise we assign $v^F$ to $J_1$ and $v^T$ to $J_2$.
	Hence any job in $J_2$ can be processed with speed $1$ similarly for any job in $J_1$.
	Hence, $\sumcost \le \binom{n+1}{2}\frac{1}{s_1} \le \frac{49|V|^2}{s_1}$ ($\cmaxcost \le \frac{n}{s_1} \le \frac{7 |V|}{s_1}$), for an optimal schedule.
	Now we can state that it is sufficient to set $s_1 = 49 d |V|^2 + 1$ ($s_1 = 7 d |V| + 1$) to prove the theorem.
	
	On the other hand, assume that the answer for an instance of \ThreeSat is \texttt{NO}.
	Assume that there exists a schedule with $\sumcost < 1$ ($\cmaxcost < 1$).
	Assume that there is a partition, such that both $v^T$ and $v^F$ have no jobs from it assigned in the schedule, then $\sumcost \ge 1$ ($\cmaxcost \ge 1$), a contradiction.
	Thus assume then, that $j_c \in J_1$ is a job assigned to a machine $m$ with $s(j_c, m) = 1$; in this case we also clearly have a contradiction.
	Hence, each $j_c \in J_1$ is assigned to a machine corresponding to a valuation of the variable fulfilling $c$, hence there exists a fulfilling valuation, a contradiction.
	Hence for such an instance for any schedule $\sumcost \ge 1$ ($\cmaxcost \ge 1$).
	
	Clearly, by using this $d$-approximate algorithm on an instance of the scheduling problem corresponding to an \texttt{YES} instance of \ThreeSat we would be able to obtain a schedule with $\sumcost < 1$ ($\cmaxcost < 1$).
	On the other hand, for an instance corresponding to a \texttt{NO} instance there is no schedule with $\sumcost \le 1$ ($\cmaxcost \le 1$).
\end{proof}

\bibliographystyle{elsarticle-harv}
\bibliography{ipdps}
\end{document}